\documentclass[onecolumn]{article} \usepackage[pdftex]{graphicx}
\usepackage{amsmath,amsfonts}                   
\usepackage{bm,amsbsy}  \usepackage{mathrsfs} \usepackage{eucal} \usepackage{url}
  
\usepackage[protrusion=true,expansion=true]{microtype} \usepackage{xspace}

\usepackage{cite} 
\allowdisplaybreaks[1]

\usepackage{jmlr2e}

\graphicspath{{figs/}}

\newcommand{\Laplace}{\mathcal{L}}
\DeclareMathOperator*{\var}{var}
\newcommand{\vpi}{{\bm{\pi}}}
\newcommand{\vpisamp}{{\bm{\pi}}^*}

\newcommand{\etal}{{\it et al.}\xspace}
\newcommand{\deriv}[2]{\frac{d #1}{d #2}}
\newcommand{\E}{\mathbb{E}}
\newcommand{\D}[2]{\frac{\partial #1}{\partial #2}}

\newcommand{\inv}{^{-1}}

\newcommand{\vm}{{\mathbf{m}}}  \newcommand{\vn}{{\mathbf{n}}} \newcommand{\vx}{{\mathbf{x}}} \newcommand{\vp}{\mathbf{p}} \newcommand{\nmax}{M}
\newcommand{\valpha}{\vec{\alpha}}  \newcommand{\Beta}{\mathrm{Beta}}
\newcommand{\Dir}{\mathrm{Dir}}
\newcommand{\alphabetX}{\mathcal{X}}
 
\newcommand{\alphabetSize}{\ensuremath{\mathcal{A}}}
                        
\newcommand{\ind}[1]{\mathbf{1}_{\left\{ #1\right\}}}                                 \newcommand{\indnobr}[1]{\mathbf{1}_{ #1}}                                 \newcommand{\Hnsb}{\hat H_{nsb}}
\newcommand{\Hnsbinf}{\hat H_{ansb}}

\newcommand{\Hplug}{\hat H_{\text{plugin}}}
\newcommand{\Hpym}{\hat H_{\text{PYM}}}
\newcommand{\PY}{{\mbox{PY}}\xspace} \newcommand{\DP}{{\mbox{DP}}\xspace} \newcommand{\PYM}{\ensuremath{PY\!M}} \newcommand{\digam}{{\psi_0}}

\newcommand{\dm}[1]{\ensuremath{\,\mathrm{d}{#1}}}  \newcommand{\trp}{^\top} \newcommand{\MAP}{\text{MAP}}

\newcommand{\figref}[1]{Fig.~\ref{fig:#1}}

\newcommand{\ConParFinite}{{a}}
\newcommand{\ConParInfinite}{{\alpha}}

\renewcommand{\cite}{\citep}

\title{Bayesian Entropy Estimation for Countable Discrete Distributions}

\author{\name Evan Archer$^\ast$
  \email earcher@utexas.edu\\
  \addr Max Planck Institute for Biological Cybernetics,  T\"{u}bingen, Germany
  \\
  \\
  \addr Center for Perceptual Systems\\
  The University of Texas at Austin, Austin, TX 78712, USA
      \AND 
  \name Il Memming Park\thanks{EA and IP contributed equally.}
  \email memming@austin.utexas.edu\\
  \addr Center for Perceptual Systems\\
  The University of Texas at Austin, Austin, TX 78712, USA
  \AND
  \name Jonathan W. Pillow
  \email pillow@utexas.edu\\
  \addr 
  Department of Psychology, Section of Neurobiology,\\
  Division of Statistics and Scientific Computation, and Center for Perceptual Systems\\
  The University of Texas at Austin, Austin, TX 78712, USA
}

\editor{} 
\jmlrheading{}{2014?}{}{4/13, 2013}{?}{Evan Archer, Il Memming Park, and Jonathan W. Pillow}

\ShortHeadings{Bayesian Entropy Estimation for Countable Discrete Distributions}{Archer, Park, and Pillow}
\firstpageno{1}

\begin{document}
\maketitle

\begin{abstract}
  We consider the problem of estimating Shannon's entropy $H$ from
  discrete data, in cases where the number of possible symbols is
  unknown or even countably infinite.  The Pitman-Yor process, a
  generalization of Dirichlet process, provides a tractable prior
  distribution over the space of countably infinite discrete
  distributions, and has found major applications in Bayesian
  non-parametric statistics and machine learning. Here we show that it
  also provides a natural family of priors for Bayesian entropy
  estimation, due to the fact that moments of the induced posterior
  distribution over $H$ can be computed analytically. We derive
  formulas for the posterior mean (Bayes' least squares estimate) and
  variance under Dirichlet and Pitman-Yor process priors.  Moreover,
  we show that a fixed Dirichlet or Pitman-Yor process prior implies a
  narrow prior distribution over $H$, meaning the prior strongly
  determines the entropy estimate in the under-sampled regime. We
  derive a family of continuous mixing measures such that the
  resulting mixture of Pitman-Yor processes produces an approximately
  flat prior over $H$.  We show that the resulting Pitman-Yor Mixture
  (PYM) entropy estimator is consistent for a large class of
  distributions. We explore the theoretical properties of the resulting estimator, and show that it performs well both in simulation and in application to real data.
\end{abstract}

\begin{keywords}
entropy,
information theory,
Bayesian estimation,
Bayesian nonparametrics,
Dirichlet process,
Pitman--Yor process,
neural coding
\end{keywords}
   
\section{Introduction} 
Shannon's discrete entropy appears as a basic statistic in many
fields, from probability theory to engineering and even ecology and neuroscience.
While entropy may best be known as a theoretical quantity, its
accurate estimation from data is an important part of many applications.
For example, entropy is employed in the study of information processing in
neuroscience, where the coding of neurons is typically unknown~\cite{Strong98a, Barbieri04, Shlens07, Rolls1999}. Entropy estimates are used to quantify the coding properties of a neural systems, such as the channel capacity, or form a step in computing other information-theoretic quantities, such as the mutual information.  
Entropy is also used in statistics and machine learning for estimating dependency
structure and inferring causal relations~\cite{Chow1968,
Hlavackova-Schindler2007}, for example in molecular biology~\cite{Hausser09};
as a tool in the study of complexity and dynamics in physics
\cite{Letellier06}; and as a measure of diversity in ecology
\cite{Chao03} and genetics~\cite{Farach95}.
Each of these studies, confronted with data arising
from an unknown discrete distribution, seeks to estimate the entropy
rather than the distribution itself. The reason for this, fundamentally, is the difficulty of density estimation: it may not be feasible to collect enough data to usefully estimate the full distribution. The problem is not just that we may not have enough data to estimate the frequency of an event accurately. In the so-called ``undersampled regime'', we may not even observe all events that have non-zero frequency. In general, estimating a density in this setting is a hopeless endeavor. 

Estimating the entropy is much easier than estimating the full distribution. In fact, in many cases, entropy can be accurately estimated with fewer samples than the number of distinct symbols. Nevertheless, entropy estimation remains a difficult
problem: there is no unbiased estimator for entropy, and the maximum
likelihood estimator is severely biased for small datasets \cite{Paninski03}.  Many
previous studies have focused upon methods for computing and reducing
this bias
\cite{Miller55,Panzeri96,Strong98a,Paninski03,Grassberger08}. In this
paper we instead take a Bayesian approach, building upon the work of
\citet{Nemenman02}. Our basic strategy is to place a prior over the
space of discrete probability distributions, and then perform
inference using the induced posterior distribution over entropy. (See
\figref{graphicalmodel}).

We focus on the under-sampled regime, where the number of unique
symbols observed in the data is small in comparison with an unknown
(perhaps infinite) number of possible symbols.  While the true
distribution underlying any real dataset is undoubtedly finite,
formulating a model on the space of infinite-dimensional distributions
allows us to be agnostic about the true cardinality.  This assumption
leads to a convenient and tractable estimator that converges to the
true entropy even when the maximal support of the distribution is
known to be finite.  As an example, consider an arbitrary distribution
over $N$-grams on a binary alphabet.  While there are at most $2^N$
symbols, the true cardinality of the response distribution is
generally unknown to the experimenter.  A given distribution may be
supported on any number of symbols less than $2^N$.  In some
applications, there may not even be a known upper bound on the number
of symbols (e.g., the problem of estimating the number of undiscovered
species in a given region).  By using a prior with support for
countably infinite distributions, our model does not require {\it a
  priori} specification of the alphabet size.

The Pitman-Yor process (PYP), a two-parameter generalization of the
Dirichlet process (DP) \cite{Pitman97,Ishwaran03,Goldwater06},
provides an attractive family of priors in this setting, since: (1)
the posterior distribution over entropy has analytically tractable
moments; and (2) distributions drawn from a PYP can exhibit
power-law tails, a feature commonly observed in data from social,
biological, and physical systems~\cite{Zipf49,DudokDeWit1999,Newman05}.

We show that a PYP prior with fixed hyperparameters imposes a narrow
prior distribution over entropy, leading to severe bias and overly
narrow posterior credible intervals given a small dataset.  Our
approach, inspired by~\citet{Nemenman02}, is to introduce a family of
mixing measures over Pitman-Yor processes such that the resulting
Pitman-Yor Mixture (PYM) prior provides an approximately
non-informative (i.e., flat) prior over entropy.

The remainder of the paper is organized as follows.  In Section
\ref{sec:basics}, we introduce the entropy estimation problem and
review prior work. In Section 3, we introduce the Dirichlet and
Pitman-Yor processes and discuss key mathematical properties relating
to entropy. In Section 4, we introduce a novel entropy estimator based
on PYM priors and derive several of its theoretical properties.  In
Section 5, we show compare various estimators with applications to data.

\begin{figure}[t]
\centering
\includegraphics[width=2in]{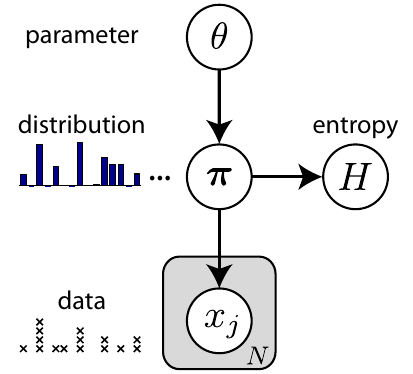}
\caption{Graphical model illustrating the ingredients for Bayesian
  entropy estimation.  Arrows indicate conditional dependencies
  between variables, and the gray ``plate'' denotes multiple copies of
  a random variable (with the number of copies $N$ indicated at
  bottom).  For entropy estimation, the joint probability distribution
  over entropy $H$, data $\vx=\{x_j\}$, discrete distribution $\vpi =
  \{\pi_i\}$, and parameter $\theta$ factorizes as:
  $p(H,\vx,\vpi,\theta) =
  p(H|\vpi)p(\vx|\vpi)p(\vpi|\theta)p(\theta)$.  Entropy is a
  deterministic function of $\vpi$, so $p(H|\vpi) = \delta(H-\sum_i
  \pi_i \log \pi_i).$ The Bayes least squares estimator corresponds to
  the posterior mean: $\E[H|\vx] = \iint p(H|\vpi) p(\vpi,\theta|\vx)
  d \vpi\, d\theta$.  }
\label{fig:graphicalmodel}
\end{figure}

\section{Entropy Estimation} \label{sec:basics}
Consider samples $\vx:=\left\{ x_j \right\}_{j=1}^{N}$ drawn
\textit{iid} from an unknown discrete distribution $\vpi := \left\{
  \pi_i \right\}_{i=1}^{\alphabetSize}$, $p(x_j = i) = \pi_i$, on a finite or (countably)
infinite alphabet $\alphabetX$ with cardinality $\alphabetSize$.
We wish to estimate the entropy of $\vpi$,
\begin{equation}
  \label{eq:entropy}
  H(\vpi) = -\sum_{i=1}^{\alphabetSize} \pi_i \log \pi_i.
\end{equation}
We are interested in the ``undersampled regime'', $N \ll \alphabetSize$, where many of the symbols remain unobserved. We will see that a naive approach to entropy estimation in this regime results in severely biased estimators, and briefly review approaches for correcting this bias. 
We then consider Bayesian techniques for entropy estimation in general before introducing the Nemenman--Shafee--Bialek (NSB) method upon which the remainder of the article will build. 

\subsection{Plugin estimator and bias-correction methods}

Perhaps the most straightforward entropy estimation technique is to
estimate the distribution $\vpi$ and then use the plugin formula
\eqref{eq:entropy} to evaluate its entropy.  The empirical
distribution $\hat \vpi = (\hat \pi_1, \ldots,
\hat\pi_{\alphabetSize})$ is computed by normalizing the observed
counts $\vn:=(n_1, \ldots, n_{\alphabetSize})$ of each symbol,
\begin{equation}
\hat \pi_k =n_k/N, \qquad n_k =
\sum_{i=1}^N\ind{x_i=k},
\end{equation}
for each $k \in \alphabetX$. Plugging this estimate for $\vpi$ into
\eqref{eq:entropy}, we obtain the so-called ``plugin'' estimator:
\begin{equation}
 \label{eq:Hplugin}
\Hplug = -\sum \hat \pi_i \log \hat \pi_i,
\end{equation}
which is also the maximum-likelihood estimator under categorical (or
multinomial) likelihood.

Despite its simplicity and desirable asymptotic properties, $\Hplug$
exhibits substantial negative bias in the undersampled regime. There
exists a large literature on methods for removing this bias, much of
which considers the setting in which $\alphabetSize$ is known and
finite. One popular and well-studied method involves taking a series
expansion of the bias~\cite{Miller55,Treves95,Panzeri96,Grassberger08}
and then subtracting it from the plugin estimate. Other recent
proposals include minimizing an upper bound over a class of linear
estimators~\cite{Paninski03}, and a James-Stein
estimator~\cite{Hausser09}. Recent work has also considered countably
infinite alphabets. The coverage-adjusted estimator
(CAE)~\cite{Chao03,Vu07} addresses bias by combining the
Horvitz-Thompson estimator with a nonparametric estimate of the proportion of total probability
mass (the ``coverage'') accounted for by the observed data $\vx$. In a
similar spirit, \citet{Zhang2012} proposed an estimator based on the
Good-Turing estimate of population size.

\subsection{Bayesian entropy estimation}\label{sec:bayesian:estimation}
The Bayesian approach to entropy estimation involves formulating a
prior over distributions $\vpi$, and then turning the crank of
Bayesian inference to infer $H$ using the posterior distribution.
Bayes' least squares (BLS) estimators take the form:
\begin{equation}\label{eq:basic:Bayesian}
  \hat H(\vx) = \E[H|\vx] = \int H(\vpi) p(H|\vpi) p(\vpi|\vx)\, \mathrm{d}\vpi,
\end{equation}
where $p(\vpi|\vx)$ is the posterior over $\vpi$ under some prior
$p(\vpi)$ and discrete likelihood $p(\vx|\vpi)$, and 
\begin{equation}
p(H|\vpi) =
\delta(H + \sum_i \pi_i \log \pi_i), \end{equation}
since $H$ is deterministically related to $\vpi$. To the extent that
$p(\vpi)$ expresses our true prior uncertainty over the unknown
distribution that generated the data, this estimate is optimal (in a
least-squares sense), and the corresponding credible intervals capture
our uncertainty about $H$ given the data.

For distributions with known finite alphabet size $\alphabetSize$, the
Dirichlet distribution provides an obvious choice of prior due to its
conjugacy with the categorical distribution.  It takes the form
\begin{equation}
  p_{Dir}(\vpi) \propto \prod_{i=1}^{\alphabetSize} \pi_i^{\ConParFinite-1},
\end{equation}
for $\vpi$ on the ${\alphabetSize}$-dimensional simplex ($\pi_i\geq
1$, $\sum \pi_i = 1$), where $\ConParFinite>0$ is a ``concentration'' parameter
\cite{Hutter02}.
Many previously proposed estimators can be viewed as Bayesian under a
Dirichlet prior with particular fixed choice of $\ConParFinite$.
See \citet{Hausser09}
for a historical overview of entropy estimators arising from specific
choices of $\ConParFinite$.

\subsection{Nemenman-Shafee-Bialek (NSB) estimator} \label{sec:NSB} 

In a seminal paper, \citet{Nemenman02} showed that for
finite distributions with known $\alphabetSize$, Dirichlet priors with
fixed $\ConParFinite$ impose a narrow prior distribution over entropy.
In the undersampled regime, Bayesian estimates based on such highly informative priors are essentially determined by the value of $\ConParFinite$.
Moreover, they have undesirably narrow posterior credible intervals, reflecting narrow prior uncertainty rather than
strong evidence from the data. (These estimators generally give
incorrect answers with high confidence!).  To address this problem, \citet{Nemenman02} suggested a
mixture-of-Dirichlets prior:
\begin{equation} \label{eq:Dir:prior} p(\vpi) = \int
  p_{\Dir}(\vpi|\ConParFinite) p(\ConParFinite) \dm{\ConParFinite},
\end{equation}
where $p_\Dir(\vpi|\ConParFinite)$ denotes a $\Dir(\ConParFinite)$ prior on $\vpi$,
and $p(\ConParFinite)$ denotes a set of mixing weights, given by
\begin{equation}\label{eq:Dir:wts}
p(\ConParFinite) \propto \frac{d}{d \ConParFinite} \E[H|\ConParFinite]
    = \alphabetSize \psi_1(\alphabetSize \ConParFinite+1) - \psi_1(\ConParFinite+ 1),
\end{equation}
where $\E[H|\ConParFinite]$ denotes the expected value of $H$ under a
$\Dir(\ConParFinite)$ prior, and $\psi_1(\cdot)$ denotes the tri-gamma
function.  To the extent that $p(H|\ConParFinite)$ resembles a delta
function, \eqref{eq:Dir:prior} and \eqref{eq:Dir:wts} imply a uniform
prior for $H$ on $[0,\log \alphabetSize]$. The BLS estimator under the
NSB prior can be written:
\begin{align}\label{eq:HhatNSB}
  \Hnsb &= \E[H|\vx] = \iint H(\vpi) p(\vpi|\vx,\ConParFinite)\, p(\ConParFinite|\vx) \dm{\vpi} \dm{\ConParFinite}
  \nonumber\\
      &=  \int \E[H|\vx,\ConParFinite] \frac{p(\vx|\ConParFinite) p(\ConParFinite)}{p(\vx)}\,\dm{\ConParFinite},
\end{align}
where $E[H|\vx,\ConParFinite]$ is the posterior mean under a $\Dir(\ConParFinite)$
prior, and $p(\vx|\ConParFinite)$ denotes the evidence, which has a P\'olya
distribution~\cite{Minka03}:
\begin{align}
p(\vx|\ConParFinite)
    &= \int p(\vx|\vpi)p(\vpi|\ConParFinite) \dm{\vpi}
    \nonumber\\
    &= \frac{(N!)
  \Gamma(\alphabetSize\ConParFinite)}{\Gamma(\ConParFinite)^\alphabetSize \Gamma ( N + \alphabetSize\ConParFinite)} 
\prod_{i=1}^\alphabetSize \frac{\Gamma(n_i + \ConParFinite)}{ n_i!}.
\end{align}

The NSB estimate $\Hnsb$ and its posterior variance are fast to
compute via 1D numerical integration in $\ConParFinite$ using closed-form
expressions for the first two moments of the posterior distribution of
$H$ given $\ConParFinite$.
The forms for these moments are
discussed in~\citet{Wolpert95,Nemenman02}, but the full formulae are
not explicitly shown.  Here we state the results:
\begin{align}
  \E[H|\vx,\ConParFinite] &= \psi_0(\tilde N+1) - \sum_i \frac{\tilde
    n_i}{\tilde N} \psi_0(\tilde n_i+ 1)
    \label{eq:Dir:posterior:mean}
    \\
   \E[H^2|\vx,\ConParFinite] &= 
	\sum_{i\neq k}
	\frac{ \tilde n_i \tilde n_k }{(\tilde N+1)\tilde N}
	    I_{i,k}
	+
	\sum_i 
	\frac{(\tilde n_i+1)\tilde n_i}{(\tilde N+1)\tilde N}
	    J_i
    \label{eq:Dir:posterior:var}
	  \\
	  I_{i,k} &=
	\left(
	  \psi_0(\tilde n_k+1)-\psi_0(\tilde N + 2)
	\right) 
	\left(
	    \psi_0(\tilde n_i+1)-\psi_0(\tilde N + 2)
	\right)  -\psi_1(\tilde N + 2) 
    \nonumber \\
	J_i &= 
	(\psi_0( \tilde n_i + 2) - 
	  \psi_0( \tilde N + 2))^2 +
	  \psi_1( \tilde n_i + 2) - \psi_1(\tilde N + 2),
    \nonumber
\end{align}
where $\tilde n_i = n_i + \ConParFinite$ are counts plus prior
``pseudocount'' $\ConParFinite$, $\tilde N = \sum \tilde n_i$ is the total of
counts plus pseudocounts, and $\psi_n$ is the polygamma of $n$-th
order (i.e., $\psi_0$ is the digamma function).  Finally,
$\var[H|\vn,\ConParFinite] = \E[H^2|\vn,\ConParFinite] - \E[H|\vn,\ConParFinite]^2$. We
derive these formulae in the Appendix, and in addition provide an
alternative derivation using a size-biased sampling formulae discussed
in Section \ref{sec:prior:DPPY}.

\subsection{Asymptotic NSB estimator}\label{sec:ANSB}

Nemenman \etal have proposed an extension of the NSB estimator to
countably infinite distributions (or distributions with unknown
cardinality), using a zeroth order approximation to $\Hnsb$ in the
limit $\alphabetSize \to \infty$ which we refer to as
asymptotic-NSB (ANSB)  \cite{Nemenman04,Nemenman11},
\begin{equation}
\Hnsbinf = 2 \log(N) + \psi_0(N-K) - \psi_0(1) - \log(2) \label{eq:ANSB},
\end{equation}
where $K$ is the number of distinct symbols in the sample.
Note that the ANSB estimator is designed specifically for an extremely undersampled regime ($K \sim N$), which we refer to as the ``ANSB approximation regime''. The fact that ANSB diverges with $N$ in the well-sampled regime (as noted by [Vu07]) is therefore consistent with its design. In our experiments with ANSB in subsequent sections, we follow  \cite{Nemenman11} and define the ANSB approximation regime to be that region such that $E[K_N] / N > 0.9$, where $K_N$ is the number of unique symbols appearing in a sample of size $N$.

\section{Dirichlet and Pitman-Yor Process Priors}\label{sec:prior:DPPY}

To construct a prior over unknown or countably infinite discrete
distributions, we borrow tools from nonparametric Bayesian
statistics. 
The Dirichlet Process (DP) and Pitman-Yor process (PYP)
define stochastic processes whose samples are countably infinite discrete
distributions~\cite{Ferguson1973,Pitman97}. 
A sample from a DP or PYP may be written as 
$\sum_{i=1}^\infty \pi_i \delta_{\phi_i}$, where now $\vpi = \{\pi_i\}$
denotes a countably infinite set of `weights' on a set of
atoms $\{\phi_i\}$ drawn from some base probability measure, where
$\delta_{\phi_i}$ is a delta function on the atom $\phi_i$.\footnote{ Here, we will assume the base measure is non-atomic, so
    that the atoms $\phi_i$'s are distinct with probability one. This
    allows us to ignore the base measure, making entropy of the
    distribution equal to the entropy of the weights $\vpi$.}
We use DP and PYP to define a prior distribution on the
infinite-dimensional simplex.
The prior distribution over $\vpi$ under the DP or PYP is technically
called the GEM \footnote{GEM stands for ``Griffiths, Engen and McCloskey'', after three researchers who considered these ideas early on \cite{Ewens1990}.} distribution or the two-parameter Poisson-Dirichlet
distribution, but we will abuse terminology by referring to both the
process and its associated weight distribution by the same symbol, DP
or PY~\cite{Ishwaran02}.

The DP distribution over $\vpi$ results from a limit of the (finite) Dirichlet distribution where alphabet size grows and concentration parameter shrinks: $\alphabetSize \rightarrow \infty$ and $\ConParFinite\rightarrow 0$ s.t.\ $\ConParFinite \alphabetSize \rightarrow \ConParInfinite$.  
The PYP distribution over $\vpi$ generalizes the DP to
allow power-law tails, and includes DP as a special
case~\cite{Kingman75,Pitman97}.
For $\PY(d,\ConParInfinite)$ with $d \neq 0$, the tails
approximately follow a power-law: $\pi_i \propto
\left(i\right)^{-\frac{1}{d}}$~(pp.~867,~\citet{Pitman97}).\footnote{Note that the power-law exponent is given incorrectly in
  \cite{Goldwater06,Teh06b}.}  Many natural phenomena such as city
size, language, spike responses, etc., also exhibit power-law tails
\cite{Zipf49,Newman05}. Fig.~\ref{fig:powerlaw} shows two such
examples, along with a sample drawn from the
best-fitting DP and PYP distributions.

Let $\PY(d,\ConParInfinite)$ denote the PYP with {\it discount}
parameter $d$ and {\it concentration} parameter $\ConParInfinite$
(also called the ``Dirichlet parameter''), for $d\in[0,1),
\ConParInfinite>-d$.  When $d = 0$, this reduces to the Dirichlet
process, $\DP(\ConParInfinite)$.  To gain intuition for the DP and
PYP, it is useful to consider typical samples $\vpi$ with weights
$\{\pi_i\}$ sorted in decreasing order of probability, so that
$\pi_{(1)} > \pi_{(2)} > \cdots$.  The concentration parameter
$\ConParInfinite$ controls how much of the probability mass is concentrated in
the first few samples, that is, in the head instead of the tail of the
sorted distribution.  For small $\ConParInfinite$, the first few weights carry
most of the probability mass, whereas for large $\ConParInfinite$, the
probability mass is more spread out so that $\vpi$ is more uniform. As
noted above, the discount parameter $d$ controls the shape of the
tail, with larger $d$ giving heavier power-law tails, and $d=0$ giving
exponential tails.

We can draw samples $\vpi \sim \PY(d,\ConParInfinite)$
using an infinite sequence of independent Beta random variables in a
process known as ``stick-breaking''~\cite{Ishwaran01}:
\begin{equation}
  \label{eq:PYstickbreak}
\beta_i
   \sim \Beta(1-d, \ConParInfinite+id), \qquad \tilde \pi_i = \prod_{j=1}^{i-1} (1-\beta_j) \beta_i,
\end{equation}
where $\tilde \pi_i$ is known as the $i$'th {\it size-biased
permutation} from $\vpi$~\cite{Pitman96}.  The $\tilde \pi_i$ sampled
in this manner are not strictly decreasing, but decrease on average
such that $\sum_{i=1}^\infty \tilde \pi_i = 1$ with probability 1
\cite{Pitman97}.

\begin{figure}[t!]
\centering
\includegraphics[width=6in]{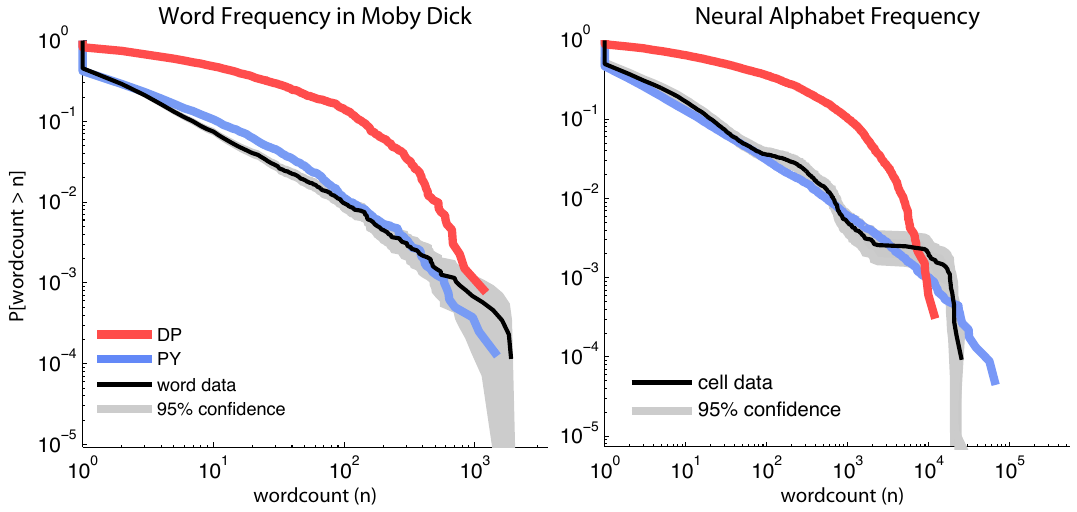}
\caption{Empirical cumulative distribution functions of words in
  natural language (above) and neural spike patterns (below).  We compare samples from the DP (red) and
  PYP (blue) priors for two datasets with heavy tails (black). In both
  cases, we compare the empirical CDF estimated from data to
  distributions drawn from DP and PYP using the ML values of $\ConParInfinite$
  and $(d,\ConParInfinite)$, respectively. For both datasets, PYP better
  captures the heavy-tailed behavior of the data. \textbf{Left:}
  Frequency of $N = 217826$ words in the novel Moby Dick by Herman
  Melville.  \textbf{Right:} Frequencies among $N = 1.2\times 10^6$
  neural spike words from $27$ simultaneously-recorded retinal
  ganglion cells, binarized and binned at 10ms.
}
\label{fig:powerlaw}
\end{figure}

\subsection{Expectations over DP and PYP priors}
A key virtue of PYP priors for our purposes is a mathematical property
called {\it invariance under size-biased sampling}, which allows us to
convert expectations over $\vpi$ on the infinite-dimensional simplex
(which are required for computing the mean and variance of $H$ given
data) into one- or two-dimensional integrals with respect to the
distribution of the first two size-biased samples
\cite{Perman92,Pitman96}.
\begin{proposition}[Expectations with first two size-biased samples]\label{prop:sizebiased:mean}
For $\vpi \sim \PY(d,\ConParInfinite)$, 
\begin{align}
\label{eq:sizebiasedexp}
\E_{(\vpi|d,\ConParInfinite)} \left[ \sum_{i=1}^\infty f(\pi_i) \right] 
    &= \E_{(\tilde \pi_1|d,\ConParInfinite)} \left[ \frac{f(\tilde
        \pi_1)}{\tilde \pi_1} \right], \\
\label{eq:sizebiasedexpp}
\E_{(\vpi|d,\ConParInfinite)} \left[ \sum_{i,j \neq i} g(\pi_i, \pi_j) \right]  
    &= \E_{(\tilde \pi_1,\tilde \pi_2|d,\ConParInfinite)} 
    \left[
      \frac{g(\tilde \pi_1,\tilde \pi_2)}{\tilde \pi_1 \tilde \pi_2} (1 - \tilde\pi_1)
    \right],
  \end{align}
where $\tilde \pi_1$ and $\tilde \pi_2$ are the first two
size-biased samples from $\vpi$.
\end{proposition}

The first result \eqref{eq:sizebiasedexp} appears in~\cite{Pitman97},
and an analogous proof can be constructed for \eqref{eq:sizebiasedexpp} (see
Appendix).  

The direct consequence of this proposition is that the first two 
moments of $H(\vpi)$ under the PYP and DP priors have closed forms \footnote{Note that \eqref{eq:PY:prior:entropy:mean} and \eqref{eq:PY:prior:entropy:var} follow from \eqref{eq:Dir:posterior:mean} and \eqref{eq:Dir:posterior:var}, respectively, under the PY limit.},
\begin{align}
\E[H|d,\ConParInfinite] &= \digam(\ConParInfinite+1) - \digam(1-d),
\label{eq:PY:prior:entropy:mean}
\\
\var [H|d,\ConParInfinite] &=
    \frac{\ConParInfinite+d}{(\ConParInfinite+1)^2 (1-d)} 
    + \frac{1-d}{\ConParInfinite+1}\psi_1(2-d) - \psi_1(2+\ConParInfinite).
\label{eq:PY:prior:entropy:var}
\end{align}
The derivation can be found in the appendix.

\subsection{Expectations over DP and PYP posteriors}
\label{sec:posterior}
A useful property of PYP priors (for multinomial observations) is that the
posterior $p(\vpi|\vx,d,\ConParInfinite)$ takes the form of a mixture of a
Dirichlet distribution (over the observed symbols) and a Pitman-Yor
process (over the unobserved symbols)~\cite{Ishwaran03}. This makes
the integrals over the infinite-dimensional simplex tractable, and as
a result we obtain closed-form solutions for the posterior mean and
variance of $H$.  Let $K$ be the number of unique symbols observed in
$N$ samples, i.e., $K=\sum_{i=1}^{\alphabetSize} \ind{n_i >
  0}\label{eq:NumTables}$\footnote{We note that the quantity $K$ has been studied in Bayesian nonparametrics in its own right, for instance to study species diversity in ecological applications \cite{Favaro2009}.}.  Further, let $\ConParInfinite_i = n_i - d$, $N =
\sum n_i$, and $A = \sum \ConParInfinite_i = \sum_in_i - Kd = N-Kd$.  Now,
following~\cite{Ishwaran02} we write the posterior as an infinite
random vector $\vpi|\vx,d,\ConParInfinite = (p_1, p_2, p_3, \dots, p_K,
p_\ast \vpi')$, where
\begin{align}\label{eq:postPY}
(p_1, p_2, \dots, p_K, p_\ast)  &\sim \Dir(n_1 - d, \dots,
n_K - d, \ConParInfinite + Kd)
\\\nonumber
\vpi' := (\pi_1, \pi_2, \pi_3, \dots) &\sim \PY(d,\ConParInfinite + Kd).
\end{align}
The posterior mean $E[H|\vx, d,\ConParInfinite]$ is given by,
\begin{equation}\label{eq:posterior:mean}
  \E[ H| \ConParInfinite, d, \vx] = \psi_0(\ConParInfinite+N+1) 
  - \frac{\ConParInfinite+Kd}{\ConParInfinite+N}\psi_0(1-d) 
  - \frac{1}{\ConParInfinite+N} \left[\sum_{i=1}^K
    (n_i - d) \psi_0(n_i - d+1)\right].
\end{equation}

The variance, $\var[H|\vx, d,\ConParInfinite]$, also has an analytic closed form which is
fast to compute. As we discuss in detail in Appendix~\ref{sec:appendix:posterior}, $\var[H|\vx, d,\ConParInfinite]$ may be
expressed in terms of the first two moments of  $p_\ast$, $\vpi$, and 
$\vp = (p_1, \ldots, p_K)$ appearing in the posterior
\eqref{eq:postPY}. Applying the law of total variance and using the
independence properties of the posterior, we find:
\begin{align}
  \label{eq:PY:posterior:variance_indraft}
 \var[ H| d,\ConParInfinite] &=
 \E_{p_{\ast}}[(1-p_\ast)^2]\var_{\vp}[H(\vp)] + \E_{p_\ast}[
 p_{\ast}^2]\var_{\vpi}[H(\vpi)]  
 \nonumber\\
 &\quad
 + \E_{p_\ast}[\Omega^2(p_\ast)]  - \E_{p_\ast}[\Omega(p_\ast)]^2,
\end{align}
where $\Omega(p_\ast) =  (1-p_\ast) \E_{\vp}\left[ H(\vp) \right] +
p_\ast\E_{\vpi}\left[H(\vpi)\right] + H(p_\ast)$, and $H(p_\ast)= -p_\ast \log(p_\ast)-(1-p_\ast)\log(1-p_\ast)$. To specify
$\Omega(p_\ast)$, we let $\mathbf{A} = \E_{\vp}\left[ H(\vp)\right]$,
$\mathbf{B} = \E_{\vpi}\left[ H(\vpi)\right]$ so that,
\begin{align*}
\E[\Omega] &=   \E_{p_\ast}[ 1-p_\ast] \E_{\vp}\left[ H(\vp) \right] +   \E_{p_\ast}[ p_\ast]\E_{\vpi}\left[
  H(\vpi)\right] + H(p_\ast),
\\
\E[\Omega^2] &= 2  \E_{p_\ast}[p_\ast H(p_\ast)][\mathbf{B}-\mathbf{A}] 
+ 2\mathbf{A} \E_{p_\ast}[ H(p_\ast)] + \E_{p_\ast}[ h^2(p_\ast) ]
\\ & + \E_{p_\ast}[ p_\ast^2]\left[\mathbf{B}^2 - 2\mathbf{A}\mathbf{B}\right]
+  2\E_{p_\ast}[p_\ast]\mathbf{A}\mathbf{B} +  \E_{p_\ast}[ (1-p_\ast)^2]\mathbf{A}^2.
\end{align*}

\section{Entropy inference under DP and PYP priors}
The posterior expectations computed in Section \ref{sec:posterior}  provide a class of entropy estimators for distributions with countably-infinite support. For each choice of $(d,\ConParInfinite)$, $\E[H|\ConParInfinite,d,\vx]$ is the posterior mean under a $PY(d,\ConParInfinite)$ prior, analogous to the fixed-$\ConParInfinite$ Dirichlet priors discussed in Section \ref{sec:bayesian:estimation}. Unfortunately, fixed $PY(d,\ConParInfinite)$ priors also carry the same difficulties as fixed Dirichlet priors. A fixed-parameter $\PY(d,\ConParInfinite)$ prior on $\vpi$ results in a highly concentrated prior distribution on entropy (\figref{Hpriorvsad}). 

We address this problem by introducing a mixture prior $p(d,\ConParInfinite)$
on $\PY(d,\ConParInfinite)$ under which the implied prior on entropy is flat.\footnote{Notice, however, that by constructing a flat prior on
    entropy, we do not obtain an objective prior.  Here, we are not
    interested in estimating the underlying high-dimensional
    probabilities $\{\pi_i\}$, but rather in estimating a single
    statistic.  An objective prior on the model parameters is not
    necessarily optimal for estimating entropy: entropy is not a
    parameter in our model.  In fact, Jeffreys' prior for multinomial
    observations is exactly a Dirichlet distribution with a fixed
    $\alpha = 1/2$. As mentioned in the text, such Bayesian priors
    are highly informative about the entropy.}
We then define the BLS entropy estimator under this mixture prior, the Pitman-Yor Mixture (PYM) estimator, and discuss some of its theoretical properties. Finally, we turn to the computation of PYM, discussing methods for sampling, and numerical quadrature integration.

\begin{figure}[t]
\centering
\includegraphics[width=\columnwidth]{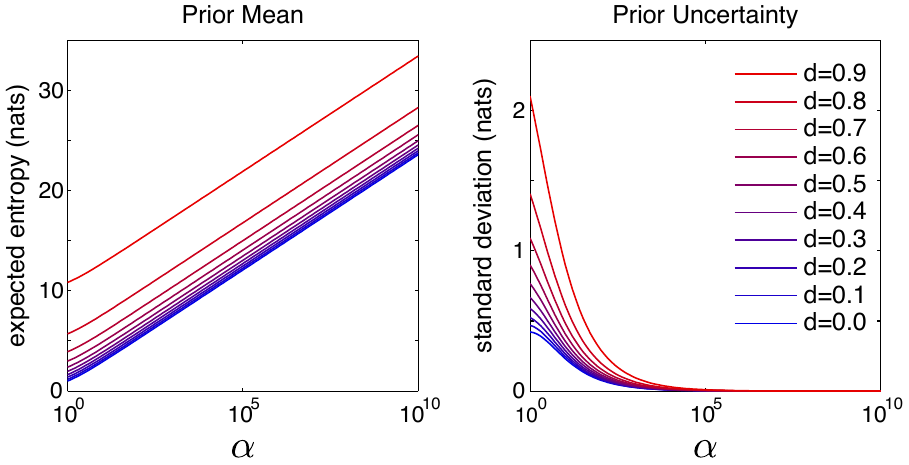}
\caption{Prior entropy mean and variance
\eqref{eq:PY:prior:entropy:mean} and \ref{eq:PY:prior:entropy:var} as a function of $\ConParInfinite$ and $d$.
Note that entropy is approximately linear in $\log \ConParInfinite$.  For large values
of $\ConParInfinite$, $p(H(\vpi)|d,\ConParInfinite)$ is highly concentrated around the mean.}
\label{fig:Hpriorvsad}
\end{figure}

\subsection{Pitman-Yor process mixture (PYM) prior}
One way of constructing a flat mixture prior is to follow the approach of \citet{Nemenman02} by setting $p(d,\ConParInfinite)$ proportional to the
derivative of the expected entropy \eqref{eq:PY:prior:entropy:mean}. 
Unlike NSB, we have two parameters through which to control the
prior expected entropy.
For instance, large prior (expected) entropies can arise either
from large values of $\ConParInfinite$ (as in the DP) or from values of $d$
near 1 (see \figref{Hpriorvsad}A).
We can explicitly control this
trade-off by reparameterizing PYP as follows,
\begin{equation}
  \label{eq:reparameterization}
  h = \digam(\ConParInfinite+1)-\digam(1-d),
  \quad
  \gamma = \frac{\digam(1)-\digam(1-d)}{\digam(\ConParInfinite+1)-\digam(1-d)},
\end{equation}
where $h>0$ is equal to the expected prior entropy
\eqref{eq:PY:prior:entropy:mean}, and $\gamma \in [0,\infty) $ captures
prior beliefs about tail behavior (\figref{pyparam}A).  For
$\gamma=0$, we have the DP (i.e., $d=0$, giving $\vpi$ with
exponential tails), while for $\gamma=1$ we have a $\PY(d,0)$ process
(i.e., $\ConParInfinite=0$, yielding $\vpi$ with power-law tails). In the limit where $\alpha\to -1$ and $d\to 1$, $\gamma \to \infty$. Where
required, the inverse transformation to standard PY parameters is
given by: $ \ConParInfinite = \digam\inv \left(h(1-\gamma)+\digam(1)\right) -
1$, $d = 1-\digam\inv\left(\digam(1)-h \gamma\right),$ where
$\digam\inv(\cdot)$ denotes the inverse digamma function.

We can construct an (approximately) flat improper distribution over $H$ on
$[0,\infty]$ by setting $p(h,\gamma) = q(\gamma)$ for all $h$, where
$q$ is any density on $[0,\infty)$. We call this the Pitman-Yor
process mixture (PYM) prior.
The induced prior on entropy is thus:
\begin{equation}
    p (H) = \iint p(H|\vpi) p(\vpi| \gamma,h) p(\gamma,h)
    \dm{\gamma}\dm{h}, 
\end{equation}
where $p(\vpi| \gamma,h)$ denotes a PYP on $\vpi$ with parameters
$\gamma,h$. We compare only three choices of $q(\gamma)$ here. However, the prior $q(\gamma)$ is not fixed, but may be adapted to reflect prior beliefs about the dataset at hand.  A $q(\gamma)$ that places probability mass on larger $\gamma$ (near $1$) results in a prior that prefers heavy-tailed behavior and high entropy, whereas weight on small $\gamma$ prefers exponential-tailed distributions. As a result, priors with more mass on large $\gamma$ will also tend to yield wider credible intervals and higher estimates of entropy. PYM mixture priors resulting from different choices of $q(\gamma)$ are all approximately flat on $H$, but each favors distributions with different tail behavior; the ability to select $q(\gamma)$ greatly enhances the flexibility of PYM, allowing the practitioner to adapt it to her own data. 

Fig.~\ref{fig:pyparam}B shows samples from this prior under three
different choices of $q(\gamma)$, for $h$ uniform on $[0,3]$. For the
experiments, we use $q(\gamma) = \exp(-\frac{10}{1-\gamma})
\mathbf{1}_{\{\gamma < 1\}}$ which yields good results by weighting
less on extremely heavy-tailed distributions\footnote{In particular,
  the restriction $\gamma < 1$ omits the corner $d\rightarrow1$ and
  $\alpha\rightarrow -d$. In this region, one can obtain arbitrarily
  large prior variance over $H$ for a given mean. However, such priors
  have very heavy tails and seem poorly-suited to data with finite or
  exponential tails, and we therefore do not explore them further
  here.}. Combined with the likelihood, the posterior $p(d,
\ConParInfinite|\vx) \propto p(\vx|d, \ConParInfinite) p(d,
\ConParInfinite)$ quickly concentrates as more data are given, as
demonstrated in Fig. \ref{fig:posterior:evidence:convergence}.

\begin{figure}[t]
\centering
\includegraphics[width=4in]{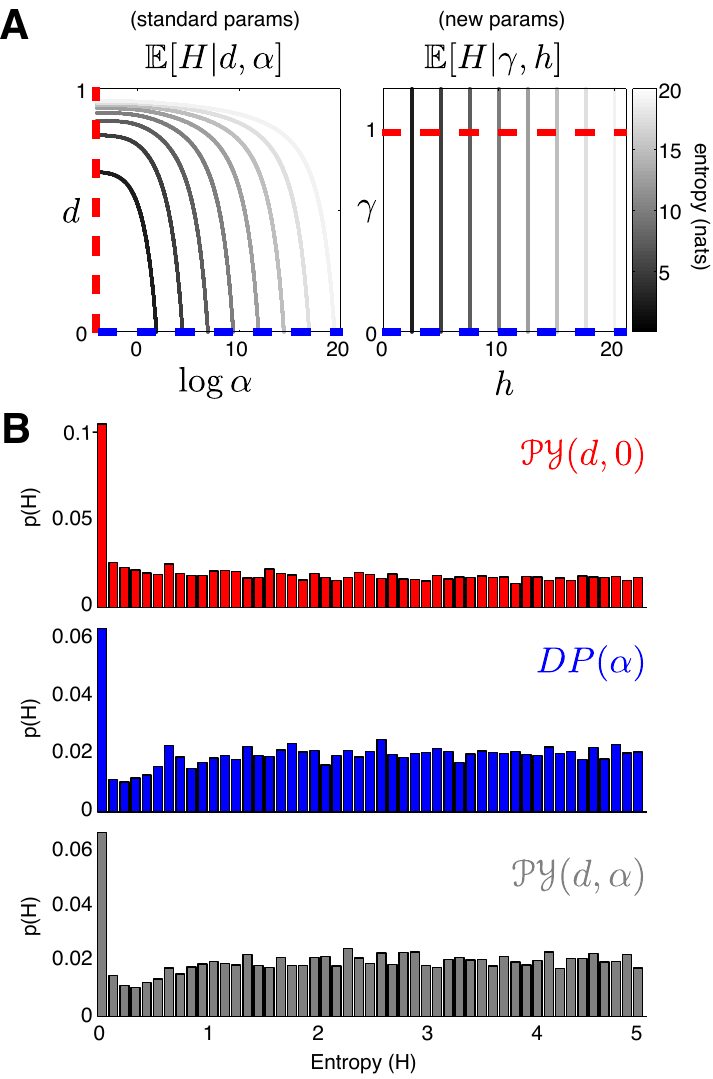}
\caption{Prior over expected entropy under Pitman-Yor process prior.
  \textbf{(A)} Left: expected entropy as a function of the natural parameters 
  $(d,\ConParInfinite)$.  Right: expected entropy as a function of transformed
  parameters $(h,\gamma)$. 
  \textbf{(B)} Sampled prior distributions ($N=5e3$) over
  entropy implied by three different PYM priors over $\alpha$ and $d$. Each PYM prior is a mixture of a different PYP: $\PY(d,0)$ (red),
  $\PY(d, \ConParInfinite)$ (grey), and $\PY(0, \ConParInfinite) = \DP(\ConParInfinite)$
  (blue).  Note that the ``true'' $p(H)$ is an improper prior supported on $[0,\infty)$. We visualize the sampled distributions only on the range from $0$ to  $5$ nats, since sampling from $\PY$ becomes prohibitively expensive with increasing expected entropy (especially as $d\to 1$).}

\label{fig:pyparam} 
\end{figure}

\begin{figure*}[t!h!p]
    \centering
    \includegraphics[width=5in]{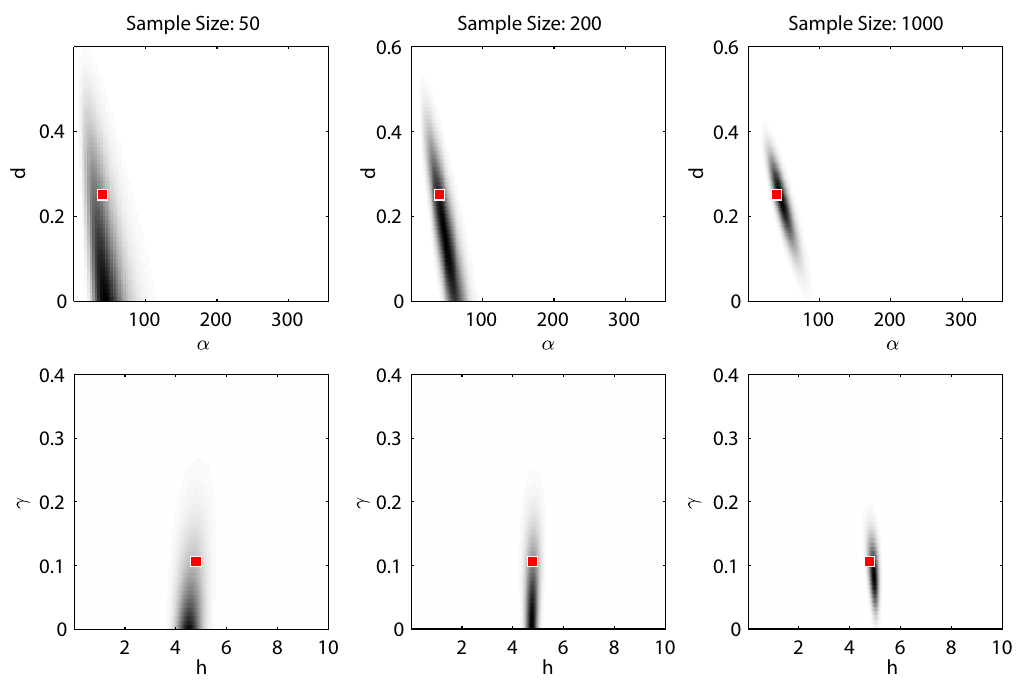}
    \caption{
	Convergence of $p(d,\ConParInfinite|\vx)$ for increasing sample size.
	Both parameterization $(d, \ConParInfinite)$ and $(\gamma, h)$ are shown.
	Data are simulated from a $\PY(0.25, 40)$ whose parameters are indicated by the red dot.
    }
    \label{fig:posterior:evidence:convergence}
\end{figure*}

\subsection{The Pitman-Yor Mixture Entropy Estimator}
Now that we have determined a prior on the infinite simplex, we turn
to the problem of inference given observations $\vx$. The Bayes least squares entropy estimator under the mixture prior $p(d,\ConParInfinite)$, the Pitman-Yor Mixture (PYM) estimator, takes the form
\begin{equation}\label{eq:HhatPYM}
\Hpym = \E[H|\vx] = \int 
    \E[H|\vx,d,\ConParInfinite] \frac{p(\vx|d,\ConParInfinite)
    p(d,\ConParInfinite)}{p(\vx)}\dm{(d,\ConParInfinite)},
\end{equation}
where $\E[H|\vx,d,\ConParInfinite]$ is the expected posterior entropy for a fixed $(d,\ConParInfinite)$ (see Section \ref{sec:posterior}). The quantity $p(\vx|d,\ConParInfinite)$ is the evidence, given by 
\begin{equation}\label{eq:evidence}
   p(\vx|d,\ConParInfinite) =
	\frac{
	    \left(
		\prod_{l=1}^{K-1} (\ConParInfinite + ld)
	    \right)
	    \left(
		\prod_{i=1}^K \Gamma(n_i - d)
	    \right)
	    \Gamma(1 + \ConParInfinite)
	}{
	    \Gamma(1 - d)^K
	    \Gamma(\ConParInfinite + N)
	}.
\end{equation}  

We can obtain posterior credible intervals for $\Hpym$ by estimating the posterior
variance $\E[(H-\Hpym)^2|\vx]$. The estimate takes the same form as
\eqref{eq:HhatPYM}, except that we replace $\E[H|\vx, d,\ConParInfinite]$ with
$\var[H|\vx, d,\ConParInfinite]$ in the integrand. 

\subsection{Computation}\label{sec:computation}
Due to the improperness of the prior $p(d,\ConParInfinite)$ and the requirement
of integrating over all $\ConParInfinite>0$ (eq.~\eqref{eq:HhatPYM}), it is not
obvious that the PYM estimate $\Hpym$ is computationally
tractable. In this section we discuss techniques for
efficient and accurate computation of $\Hpym$. First, we outline a
compressed data representation we call the ``multiplicities''
representation, which substantially reduces computational cost. Then,
we outline a fast method for performing the numerical integration over
a suitable range of $\ConParInfinite$ and $d$.

\subsubsection{Multiplicities}
\label{sec:multiplicities}
Computation of the expected entropy $\E[H|\vx,d,\ConParInfinite]$ can
be carried out more efficiently using a representation in terms of
\textit{multiplicities}, a compressed statistic often used under other
names (for instance the {\it empirical histogram distribution
    function} \cite{Paninski03}). Multiplicities are the number of
symbols that have occurred with a given frequency in the sample.
Letting $m_k = |\{i : n_i = k\}|$ denote the total number of symbols
with exactly $k$ observations in the sample gives the compressed
statistic $\vm = \left[m_0,m_1, \dots, m_\nmax\right]\trp$, where
$\nmax$ is the largest number of samples for any symbol. Note that the
dot product $[0,1,\ldots, \nmax]\trp\vm = N$, is the total number of
samples.

The multiplicities representation significantly reduces the time and
space complexity of our computations for most datasets, as we need
only compute sums and products involving the number symbols with
distinct frequencies (at most $\nmax$), rather than the total number
of symbols $K$.  In practice we compute all expressions not explicitly
involving $\vpi$ using the multiplicities representation. For instance,
in terms of the
multiplicities the evidence takes the compressed form,
\begin{align*}
    &p(\vx|d,\ConParInfinite) = p(m_1, \ldots, m_\nmax|d,\ConParInfinite) 
    \\
    &\quad= 
	\frac{
	    \Gamma(1 + \ConParInfinite)
	    \prod_{l=1}^{K-1} (\ConParInfinite + ld)
	}{
	    \Gamma(\ConParInfinite + N)
	}
	\prod_{i=1}^M
	    \left(
		\frac{\Gamma(i - d)}
		    {i! \Gamma(1 - d)}
	    \right)^{m_i}
	    \frac{\nmax!}{m_i!}.
\end{align*}

\subsubsection{Heuristic for Integral Computation}
In principle the PYM integral over $\ConParInfinite$ is supported on the range
$[0,\infty)$. In practice, however, the posterior is concentrated on a
relatively small region of parameter space. It is generally unnecessary to consider the full integral over a semi-infinite domain. Instead, we select a subregion of $[0,1]\times[0,\infty)$ which supports the posterior up to $\epsilon$ probability mass. 
The posterior is unimodal in each variable $\alpha$ and $d$ separately
(see Appendix~\ref{sec:appendix:unimodal}); however, we do not have a
proof for the unimodality of the evidence.  Nevertheless, if there are
multiple modes, they must lie on a strictly decreasing line of $d$ as
a function of $\alpha$ and, in practice, we find the posterior to be unimodal. We illustrate the concentration of the evidence visually in
figure~\ref{fig:posterior:evidence:convergence}.

We compute the hessian at the MAP parameter value,
$(d_\MAP, \ConParInfinite_\MAP)$.  Using the inverse hessian as the
covariance of a Gaussian approximation to the posterior, we
select the grid which spans $\pm 6 ~\mathrm{std}$. 
We use numerical integration (Gauss-Legendre quadrature) on this region to compute the integral. When the hessian is rank-deficient (which may occur, for instance,  when the $\ConParInfinite_\MAP = 0$ or $d_\MAP = 0$), we use Gauss-Legendre quadrature to perform the integral in $d$ over $[0,1)$, but employ a Fourier-Chebyshev numerical quadrature routine to integrate $\ConParInfinite$ over $[0,\infty)$ \cite{Boyd87}.

\subsection{Sampling the full posterior over $H$}
The closed-form expressions for the conditional moments derived in the
previous section allow us to compute PYM and its variance by
2-dimensional numerical integration. PYM's posterior mean and variance
provide essentially a Gaussian approximation to the
posterior, and corresponding credible regions. However, in some situations (see \figref{sampling}),
variance-based credible intervals are a poor approximation to the true
posterior credible intervals. In these cases we may wish to
examine the full posterior distribution over $H$. We
describe methods for exactly sampling the posterior and argue that the
posterior variance provides a good approximation to the true credible
interval in most situations.

\begin{figure}[th!bp!]
\centering
\includegraphics[width = .9\linewidth]{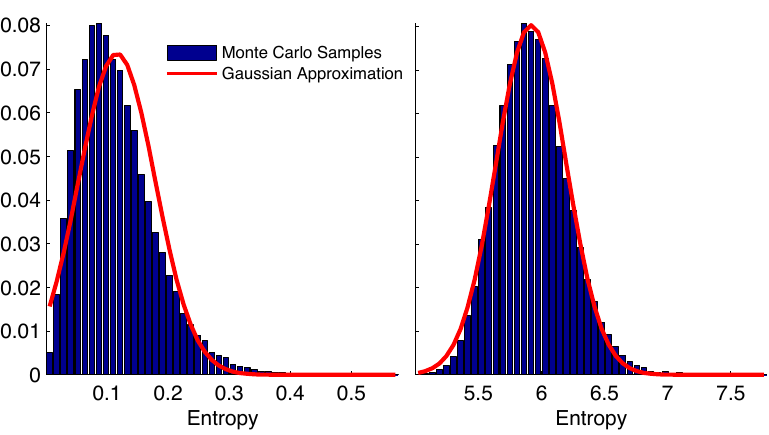}
\caption{The posterior distributions of entropy for two datasets of $100$ samples drawn from distinct distributions and the Gaussian approximation to each distribution based upon the posterior mean and variance. \textbf{(left):} Simulated example with low entropy. Notice that the true posterior is highly asymmetric, and that the Gaussian approximation does not respect the positivity of $H$.  \textbf{(right)} Simulated example with higher entropy. The Gaussian approximation is a much closer approximation to the true distribution.
}
\label{fig:sampling}
\end{figure}

 Stick-breaking, as described by \eqref{eq:PYstickbreak}, provides a straightforward algorithm for sampling
 distributions $\vpi\sim \PY(d,\ConParInfinite)$. With large enough $N_s$, stick-breaking samples $\left\{\tilde
   \pi_i\right\}_{i=1}^{N_s}$ approximate $\pi$ to arbitrary
 accuracy\footnote{Bounds on $N_s$, the number of sticks necessary to reach $\epsilon$ on average are provided in \cite{Ishwaran01}.}. Even so, sampling $\pi\sim
 \PY(d,\ConParInfinite)$ for $d$ near $1$, where $\vpi$ is likely to be
 heavy-tailed, may require intractably large $N_s$ to obtain a good approximation.   
 We address this problem by directly estimating the entropy of the tail, $\PY(d,\ConParInfinite + N_s d)$, using \eqref{eq:PY:prior:entropy:mean}. As shown in \figref{Hpriorvsad}, the prior variance of \PY becomes arbitrarily small as for large $\ConParInfinite$. For sampling, $N_s$ need only be large enough to make the variance of the tail entropy small. The resulting sample is the entropy of the (finite) samples plus the expected entropy of the tail, $H(\vpisamp) + \E[H|d,\ConParInfinite + Kd]$.\footnote{Due to the generality of the expectation formula \eqref{eq:sizebiasedexp}, this method may be applied to sample the distributions of other additive functionals of $\PY$.}

Sampling entropy is most useful for very small amounts of data drawn from distributions with low expected entropy. In \figref{posterior:evidence:convergence} we illustrate the posterior distributions of entropy in two simulated experiments. In general, as the expected entropy and sample size increase, the posterior becomes more approximately Gaussian.

\section{Theoretical properties of PYM}
Having defined PYM and discussed its practical computation, we now establish  conditions under which~\eqref{eq:HhatPYM} is defined (that is, the right--hand of the equation is finite), and also
prove some basic facts about its asymptotic properties. While $\Hpym$ is a Bayesian estimator, we wish to build connection to the literature by showing frequentist properties.

Note that the prior expectation $\E[H]$ does not exist for the improper prior defined above, since $p(h = \E[H]) \propto 1$ on $[0, \infty]$.
It is therefore reasonable to ask what conditions on the data are sufficient to obtain finite posterior expectation $\Hpym = E[H|\vx]$.
We give an answer to this question in the following short proposition (proofs of all statements may be found in the appendix),
\begin{theorem} \label{finiteposterior}
  Given a fixed dataset $\vx$ of $N$ samples, $\Hpym <\infty$ for
  any prior distribution $p(d,\ConParInfinite)$ if $N-K\geq 2$.
\end{theorem}
In other words, we require $2$ coincidences in the data for $\Hpym$ to be finite.
When no coincidences have occurred in $\vx$, we have no evidence
regarding the support of the $\vpi$, and our resulting entropy
estimate is unbounded. In fact, in the absence of coincidences, no
entropy estimator can give a reasonable estimate without prior
knowledge or assumptions about $\alphabetSize$. 

Concerns about inadequate numbers of coincidences are peculiar to the
undersampled regime; as $N\to \infty$, we will almost surely
observe each letter infinitely often. We now turn to asymptotic
considerations, establishing consistency of $\Hpym$ in the limit of
large $N$ for a broad class of distributions. It is known that the plugin is consistent for any distribution (finite or countably infinite), although the rate of convergence can be arbitrarily slow~\cite{Antos01}.
Therefore, we establish consistency by showing asymptotic convergence to the plugin estimator.

For clarity, we explicitly denote a quantity's dependence upon sample size $N$ by introducing a subscript. Thus, $\vx$ and $K$ become $\vx_N$ and $K_N$, respectively.
As a first step, we show that $\E[H|\vx_N,d,\ConParInfinite]$ converges to the plugin estimator.
\begin{theorem}\label{fixedHconv}
Assuming $\vx_N$ drawn from a fixed, finite or countably infinite
discrete distribution $\vpi$ such that $\frac{K_N}{N} \xrightarrow{P} 0$,
\begin{equation*}
      \left| 
      \E[H  |\vx_N, d, \ConParInfinite] - \E[H_\mathrm{plugin}| \vx_N] \right| \xrightarrow{P} 0
\end{equation*}
\end{theorem}
The assumption $K_N/N\to 0$ is more general than it may seem. For any
infinite discrete distribution, it holds that  $K_N \to \E[K_N]$ a.s., and
$\E[K_N]/N \to 0$ a.s.~\cite{Gnedin07}, and so $K_N/N \to 0$ in
probability for an arbitrary distribution. As a result, the right--hand--side of \eqref{eq:posterior:mean} shares its asymptotic
behavior with $\Hplug$, in particular consistency. As
\eqref{eq:posterior:mean} is consistent for each value of $\ConParInfinite$ and
$d$, it is intuitively plausible that $\Hpym$, as a mixture of such
values, should be consistent as well. However, while \eqref{eq:posterior:mean} alone is well-behaved, it is not clear that $\Hpym$ should be. 
Since $\E[H|\vx,d,\ConParInfinite] \to \infty$ as $\ConParInfinite \to \infty$, care must be taken when integrating over $p(d,\ConParInfinite|\vx)$.
Our main consistency result is, 
\begin{theorem}\label{consistencyproof}
For any proper prior or bounded improper prior $p(d, \ConParInfinite)$, if data $\vx_N$ are drawn from a fixed, countably infinite
discrete distribution $\vpi$ such that for some constant $C>0$, 
$K_N = o(N^{1-1/C})$ in probability, then
  \begin{equation*}  |\E[H | \vx_N] - \E[H_\mathrm{plugin}| \vx_N]| \xrightarrow{P} 0
\end{equation*}
\end{theorem}

Intuitively, the asymptotic behavior of $K_N/N$ is tightly related to the tail behavior of the distribution~\cite{Gnedin07}.
In particular, $K_N \sim c N^b$ with $0 < b < 1$ if and only if $\pi_i \sim c' i^{\frac{1}{b}}$ where $c$ and $c'$ are constants, and we assume $\pi_i$ is non-increasing~\cite{Gnedin07}.
The class of distributions such that $K_N = o(N^{1-1/C})$ a.s. includes the class of power-law or thinner tailed distributions, i.e., 
$\pi_i = O(i^b)$ for some $b > 1$ (again $\pi_i$ is assumed non-increasing).

We conclude this section with some remarks on the role of the prior in
Theorem~\ref{consistencyproof} as well as the significance of asymptotic
results in general. While consistency is an important property for any
estimator, we emphasize that PYM is designed to address the
undersampled regime. Indeed, since $\Hplug$ is consistent and has an optimal
rate of convergence for a large class of distributions \cite{Vu07,Antos01, Zhang2012},
asymptotic properties provide little reason to use
$\Hpym$. Nevertheless, notice that Theorem~\ref{consistencyproof} makes
very weak assumptions about $p(d,\ConParInfinite)$. In particular, the result
is not dependent upon the form of the PYM prior introduced in the
previous section: it holds for any probability distribution $p(d,
\ConParInfinite)$, or even a bounded improper prior. Thus, we can view
Theorem~\ref{consistencyproof} as a statement about a class of PYM
estimators. Almost any prior we choose on $(d, \ConParInfinite)$ results in a
consistent estimator of entropy.

\section{Simulation Results}
We compare $\Hpym$ to other proposed entropy estimators using several example datasets.
Each plot in Figs~\ref{fig:convergence:sb}, \ref{fig:convergence:powerlaw}, \ref{fig:convergence:real}, and \ref{fig:convergence:finite} shows convergence as well as small sample performance.
We compare our estimators, DPM ($d = 0$ only) and PYM ($\Hpym$), with other enumerable-support estimators:  coverage-adjusted estimator (CAE)~\cite{Chao03,Vu07}, asymptotic NSB (ANSB, section~\ref{sec:ANSB})~\cite{Nemenman11}, Grassberger's asymptotic bias correction (GR08)~\cite{Grassberger08}, and Good-Turing estimator~\cite{Zhang2012}.
Note that similar to ANSB, DPM is an asymptotic (Poisson-Dirichlet)
limit of NSB, and hence in practice behaves identically to NSB with
large but finite $K$.
We also compare with plugin~\eqref{eq:Hplugin} and a standard
Miller-Maddow (MiMa) bias correction method with a conservative
assumption that the number of uniquely observed symbols is $K$~\cite{Miller55}.
To make comparisons more straightforward, we do not apply additional bias
correction methods (e.g. jackknife) to any of the estimators.

\begin{figure}[th!bp!]
\centering
\includegraphics{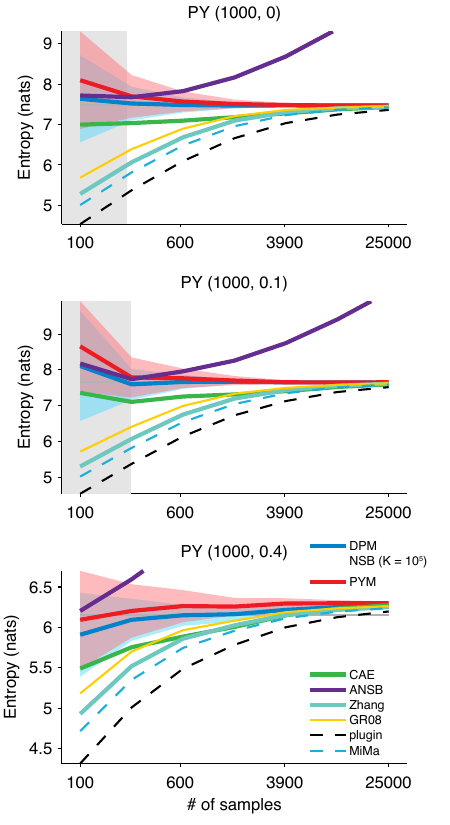}
\caption{
Comparison of estimators on stick-breaking distributions.
Poisson-Dirichlet distribution with 
$(d=0, \ConParInfinite = 1000)$ (top),
$(d = 0.1, \ConParInfinite = 1000)$ (middle),
$(d= 0.4, \ConParInfinite = 100)$ (bottom).  
Recall that the Dirichlet Process is the Pitman-Yor Process with $d=0$.
We compare our estimators (DPM, PYM) with other enumerable support estimators (CAE, ANSB, Zhang, GR08), and finite support estimators (plugin, MiMa).
Note that in these examples, the DPM estimator performs indistinguishably from NSB with alphabet size $\alphabetSize$ is fixed to a large value ($\alphabetSize=10^5$). For the experiments, we first sample a single $\vpi \sim \PY(d,\ConParInfinite)$ using the stick-breaking procedure of \eqref{eq:PYstickbreak}. 
For each $N$ ($x$-axis), we apply all estimators to each of 10 sample datasets drawn randomly from $\vpi$. 
Solid lines are averages over all 10 realizations.
Colored shaded area represent $95\%$ credible intervals averaged over all 10 realizations.
Gray shaded area represents the ANSB approximation regime defined as
expected number of unique symbols being more than 90\% the total number of samples.
}
\label{fig:convergence:sb}
\end{figure}
In each simulation, we draw 10 sample distributions $\vpi$. From each $\vpi$ we draw a dataset of N iid samples. In each figure we show the performance of all estimators averaged across the 10 sampled datasets.

The experiments of \figref{convergence:sb} show performance on a single $\vpi \sim \PY(d,\ConParInfinite)$ drawn using the stick-breaking procedure of \eqref{eq:PYstickbreak}. We draw $\pi_i$ according to \eqref{eq:PYstickbreak} in blocks of size $10^3$ until $1-\sum_{N_s} \pi_i < 10^{-3}$, where $N_s$ is the number of sticks. Unsurprisingly, PYM performs well when the data are truly generated by a Pitman-Yor process (\figref{convergence:sb}). 
Credible intervals for DPM tend to be smaller than PYM, although both shrink quickly (indicating high confidence).
When the tail of the distribution is exponentially decaying, ($d=0$ case; \figref{convergence:sb} top), DPM shows slightly improved performance.
When the tail has a strong power-law decay, (\figref{convergence:sb} bottom), PYM performs better than DPM. Most of the other estimators are consistently biased down, with the exception of ANSB.

The shaded gray area indicates the ANSB approximation regime, where the approximation used to define the ANSB estimator is approximately valid.
Although this region has no definitive boundary, it corresponds to a regime where where the average number of coincidences is small relative to the number of samples.
Following~\cite{Nemenman11}, we define the undersampled regime to be the region where $E[K_N] / N > 0.9$, where $K_N$ is the number of unique symbols appearing in a sample of size $N$.
Note that only 3 out of 10 results in Figs.~\ref{fig:convergence:sb},\ref{fig:convergence:powerlaw},\ref{fig:convergence:real},\ref{fig:convergence:finite} have shaded area; the ANSB approximation regime is not large enough to appear in the plots.
This regime appears to be designed for a relatively broad distribution (close to uniform distribution).
In cases where a single symbol has high probability, the ANSB approximation regime is essentially never valid.
In our example distributions, this is the case with for power-law distributions and $\mathcal{PY}$ distributions with large $d$.
For example, \figref{convergence:powerlaw} is already outside of the ANSB approximation regime with only 4 samples.

\begin{figure}[th!bp!]
\centering
\includegraphics{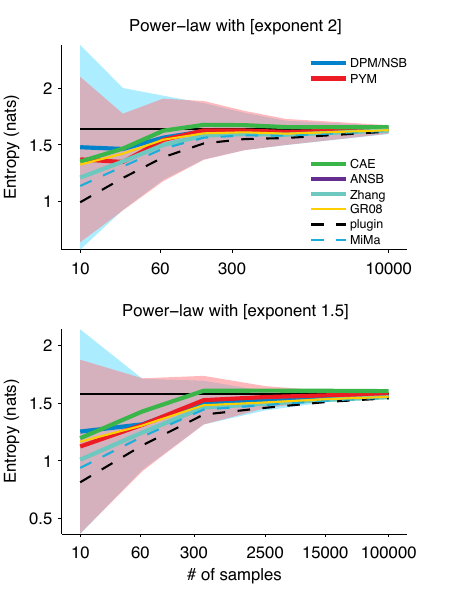}
\caption{
Comparison of estimators on power-law distributions.
The highest probabilities in those power-law distributions were big enough
such that they were never effectively undersampled.
}
\label{fig:convergence:powerlaw}
\end{figure}
Although Pitman-Yor process $\PY(d, \ConParInfinite)$ has a power-law tail controlled by $d$, the high probability portion is modulated by $\ConParInfinite$, and does not strictly follow a power-law distribution as a whole.
In \figref{convergence:powerlaw}, we evaluate the performance for $p_i \propto i^{-2}$ and $p_i \propto i^{-1.5}$.
PYM and DPM has slight negative bias, but the credible interval covers the true entropy for all sample sizes.
For small sample sizes, most estimators are negatively biased, again except for ANSB (which does not show up in the plot since it is severely biased upwards).
Notably CAE performs very well in moderate sample sizes.

\begin{figure}[thbp!]
\centering
\includegraphics{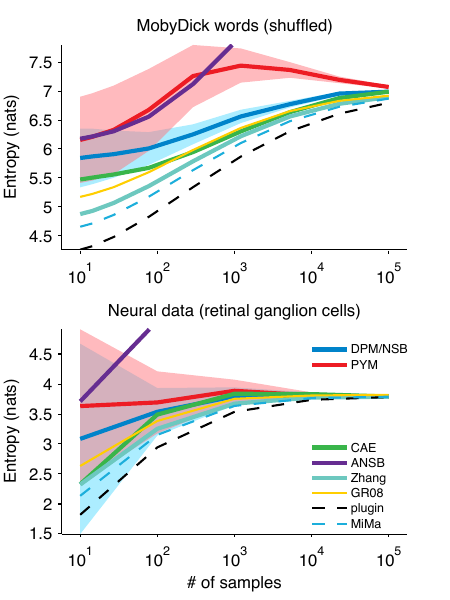}
\caption{
Comparison of estimators on real data sets.
}
\label{fig:convergence:real}
\end{figure}
In \figref{convergence:real}, we compute the entropy per word of in the novel \textit{Moby Dick} by Herman Melville, and entropy per time bin of a population of retinal ganglion cells from monkey retina~\cite{Pillow05}. We tokenized the novel into individual words using the Python library NLTK\footnote{\url{http://www.nltk.org/}}. Punctuation is disregarded. 
These real-world datasets have heavy, approximately power-law tails \footnote{We emphasize that we use the term ``power-law'' in a heuristic, descriptive sense only. We did not fit explicit power-law models to the datasets in questions, and neither do we rely upon the properties of power-law distributions in our analyses.} as pointed out earlier in \figref{powerlaw}. For Moby Dick, PYM slightly overestimates, while DPM slightly underestimates, yet both methods are closer to the entropy estimated by the full data available than other estimators. DPM is overly confident (its credible interval is too narrow), while PYM becomes overly confident with more data.
The neural data were preprocessed to be a binarized response (10 ms time bins) of 8 simultaneously recorded off-response retinal ganglion cells.
PYM, DPM, and CAE all perform well on this dataset, with both PYM and DPM bracketing the asymptotic value with their credible intervals.

Finally, we applied the denumerable support estimators to finite support distributions (\figref{convergence:finite}).
The power-law $p_n \propto n^{-1}$ has the heaviest tail among the simulations we consider, but notice that it does not define a proper distribution (the probability mass does not integrate), and so we use a truncated $1/n$ distribution with the first $1000$ symbols (\figref{convergence:finite} top).
Initially PYM shows the least bias, but DPM provides a better estimate for increasing sample size. Notice, however, that for both estimates the credible intervals consistently cover the true entropy. Interestingly, the finite support estimators perform poorly compared to DPM, CAE and PYM.
For the uniform distribution over $1000$ symbols, both DPM and PYM have slight upward bias, while CAE shows almost perfect performance (\figref{convergence:finite} middle).
For Poisson distribution, a theoretically enumerable-support distribution on the natural numbers, the tail decays so quickly that the effective support (due to machine precision) is very small ($26$ in this case). All the estimators, with the exception of ANSB, work quite well. 
The novel Moby Dick provides the most challenging data: no estimator seems to have converged, even with the full data.
Surprisingly, the Good-Turing estimator~\cite{Zhang2012} tends to perform similarly to the Grassberger and Miller-Maddow bias-correction methods.
Among such the bias-correction methods, Grassberger's method tended to show the best performance, outperforming Zhang's method.

The computation time for our estimators is $O(LG)$, where $L$ number
symbols with distinct frequencies (bounded above by the quantity $M$ defined in Section \ref{sec:multiplicities}) and $G$ is the number of gridpoints used to compute the numerical integral. Hence, computation time as a function of sample size depends upon how $L$ scales with samples size $N$, always sublinearly, and $O(N^{1/2})$ in the worse case. In our examples, computation times for $10^5$ samples are in the order of 0.1 seconds (\figref{computation-time}). Hence in practice, for the examples we have shown, more time is spent building the histogram from the data than computing the entropy estimate. 

\begin{figure}[thbp!]
\centering
\includegraphics{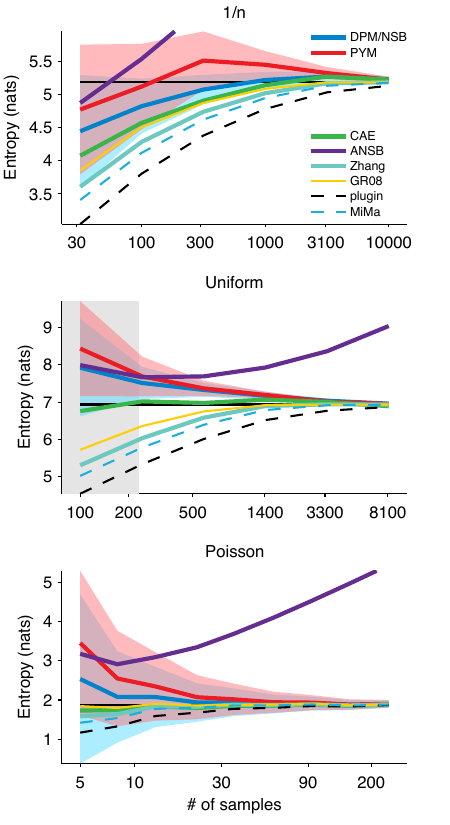}
\caption{
Comparison of estimators on finite support distributions.
Black solid line indicates the true entropy.
Poisson distribution ($\lambda = e$) has a countably infinite tail, but a very thin one---all probability mass was concentrated in 26 symbols within machine precision.
}
\label{fig:convergence:finite}
\end{figure}

\begin{figure}[thbp!]
\centering
\includegraphics{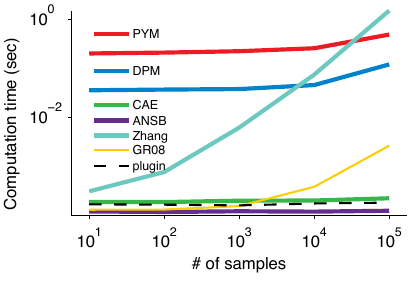}
\caption{
    Median computation time to estimate entropy for the neural data.
    The computation time excludes the preprocessing required to build the histogram and convert to multiplicity representation.
    Note that for DPM and PYM this time also includes estimating the posterior variance.
}
\label{fig:computation-time}
\end{figure}

\section{Conclusion}

In this paper we introduced PYM, a new entropy estimator for
distributions with unknown support. We derived analytic forms for the
conditional mean and variance of entropy under a \DP and \PY prior for
fixed parameters. Inspired by the work of \cite{Nemenman02}, we
defined a novel \PY mixture prior, PYM, which implies an approximately
flat prior on entropy. PYM addresses two major issues with NSB: its
dependence on knowledge of $\alphabetSize$ and its inability
(inherited from the Dirichlet distribution) to account for the
heavy-tailed distributions which abound in biological and other
natural data.  

Futher experiments on diverse datasets might reveal ways to improve PYM, such as new tactics or theory for selecting the prior on tail behavior, $q(\gamma)$. It may also prove fruitful to investigate ways to tailor PYM to a specific application, for instance by combining it with with more structured priors, such as those employed in \cite{Archer2013d}. Further, while we have shown that PYM is consistent for any prior, an expanded theory might investigate the convergence rate, perhaps in relation to the choice of prior.

We have shown that PYM performs well in comparison to other entropy estimators, and indicated its practicality in example applications to data. A MATLAB implementation of the PYM estimator is available at \url{https://github.com/pillowlab/PYMentropy}.

\appendix

\section{Derivations of Dirichlet and $\PY$ moments}
In this Appendix we present as propositions a number of technical
moment derivations used in the text.
\subsection{Mean entropy of finite Dirichlet}
\begin{proposition}[Replica trick for entropy~\cite{Wolpert95}]
  For $\vpi \sim \Dir(\alpha_1, \alpha_2, \dots,
  \alpha_{\alphabetSize})$, such that $\sum_{i=1}^{\alphabetSize}
  \alpha_i = A$, and letting $\valpha = \left( \alpha_1, \alpha_2,
    \dots, \alpha_{\alphabetSize}\right)$, we have
  \begin{eqnarray}
    \E[ H(\vpi) |\valpha] &=& \psi_0(A+1) - \sum_{i=1}^{\alphabetSize}
    \frac{\alpha_i}{A} \psi_0(\alpha_i+1)\label{eq:DirMeanEntropy}
  \end{eqnarray}
\end{proposition}

\begin{proof} First, let $c$ be the normalizer of Dirichlet, $c=\frac{
    \prod_{j} \Gamma(\alpha_j)}{ \Gamma(A)}$ and let $\Laplace$ denote the Laplace transform (on $\pi$ to $s$). Now, 
  \begin{align*}
    c \E[H|\valpha] &= \int \left( -\sum_i \pi_i \log_2
      \pi_i     \right) \delta({\scriptstyle \sum_i} \pi_i - 1) \prod_j \pi_j^{\alpha_j - 1} d\vpi
    \\
    &= -\sum_i \int \left( \pi_i^{\alpha_i} \log_2 \pi_i \right)
    \delta({\scriptstyle \sum_i} \pi_i - 1) \prod_{j \neq i}
    \pi_j^{\alpha_j - 1} d\vpi
    \\
    &= - \sum_i \int \left(
      \deriv{}{(\alpha_i)}\pi_i^{\alpha_i} \right)
    \delta({\scriptstyle \sum_i} \pi_i - 1) \prod_{j \neq i}
    \pi_j^{\alpha_j - 1} d\vpi
    \\
    &= - \sum_i \deriv{}{(\alpha_i)} \int
    \pi_i^{\alpha_i} \delta({\scriptstyle \sum_i} \pi_i - 1)
    \prod_{j \neq i} \pi_j^{\alpha_j - 1} d\vpi
    \\
    &= - \sum_i \deriv{}{(\alpha_i)} \Laplace\inv
    \left[ \Laplace(\pi_i^{\alpha_i}) \prod_{j \neq i}
      \Laplace(\pi_j^{\alpha_j - 1}) \right](1)
    \\
    &= - \sum_i \deriv{}{(\alpha_i)} \Laplace\inv
    \left[ \frac{ \Gamma(\alpha_i +1) \prod_{j \neq i}
        \Gamma({\alpha_j}) }{ s^{\sum_k (\alpha_k) + 1} }
    \right](1)
    \\
    &= - \sum_i \deriv{}{(\alpha_i)} \left[ \frac{
        \Gamma(\alpha_i +1) }{ \Gamma(\sum_k (\alpha_k) + 1) }
    \right] \prod_{j \neq i} \Gamma({\alpha_j})
    \\
    &= - \sum_i \frac{ \Gamma(\alpha_i +1) }{
        \Gamma(\sum_k \alpha_k + 1) }\left[ \psi_0(\alpha_i+ 1)  - \psi_0(A+ 1) \right]
	\prod_{j \neq i} \Gamma({\alpha_j})
    \\
    &=  \left[ \psi_0(A+1) - 
	\sum_{i=1}^{\alphabetSize} \frac{ \alpha_i }{A} 
	\psi_0(\alpha_i+ 1) \right]
    \frac{ \prod_{j} \Gamma({\alpha_j}) }{ \Gamma(A) }.
  \end{align*}
\end{proof}

\subsection{Variance entropy of finite Dirichlet}
We derive $\E[H^2(\vpi)|\valpha]$. In practice we compute
$\var[H(\vpi)|\valpha] = \E[H^2(\vpi)|\valpha] -
\E[H(\vpi)|\valpha]^2$.

\begin{proposition}
  For $\vpi \sim \Dir(\alpha_1, \alpha_2, \dots,
  \alpha_{\alphabetSize})$, such that $\sum_{i=1}^{\alphabetSize}
  \alpha_i = A$, and letting $\valpha = \left( \alpha_1, \alpha_2,
    \dots, \alpha_{\alphabetSize}\right)$, we have
  \begin{align}\label{eq:DirVarianceEntropy}
    \E[ H^2(\vpi) |\valpha] &= 
\sum_{i\neq k}
  \frac{ \alpha_i \alpha_k 
}{
  (A+ 1)(A)} 
I_{ik} 
+
\sum_i 
\frac{
  \alpha_i(\alpha_i+1)
} {
  (A + 1)(A)
}
J_i
\\ \nonumber
I_{ik} &= 
\left(
  \psi_0(\alpha_k+1)-\psi_0(A + 2)
\right) 
  \left(
    \psi_0(\alpha_i+1)
    \right. \\ &\qquad\left.-\psi_0(A + 2)
  \right) -\psi_1(A + 2) \nonumber
\\ \nonumber
J_i &=
(\psi_0( \alpha_i+2) - 
  \psi_0( A + 2))^2 +
  \psi_1(\alpha_i+2) \\&\qquad- \psi_1(A + 2) \nonumber
  \end{align}
\end{proposition}
\begin{proof}
We can evaluate the second moment in a manner similar to the mean entropy
above. First, we split the second moment into square and cross
terms. To evaluate the integral over the cross terms, we apply the
``replica trick'' twice. Letting $c$ be the normalizer of Dirichlet,
$c=\frac{ \prod_{j} \Gamma(\alpha_j)}{ \Gamma(A)}$ we have,
\begin{align*} 
&c\E[H^2|\valpha] = 
	\int \left(
	   -\sum_i \pi_i \log_2 \pi_i 	\right) ^2
	\delta({\scriptstyle \sum_i} \pi_i - 1)
	\prod_j \pi_j^{\alpha_j - 1}
	d\vpi
\\
&= 
	\sum_i \int 
        \left(
          \pi_i^2 {\log_2}^2 \pi_i
        \right)
	\delta({\scriptstyle \sum_i} \pi_i - 1)
	\prod_j \pi_j^{\alpha_j - 1}
	d\vpi \\
& + \sum_{i\neq k} \int 
        \left(
          \pi_i \log_2 \pi_i
        \right)
        \left(
          \pi_k \log_2 \pi_k
        \right)
	\delta({\scriptstyle \sum_i} \pi_i - 1)
	\prod_j \pi_j^{\alpha_j - 1}
	d\vpi 
\\
&= 
	\sum_i \int 
         \pi_i^{\alpha_i+1} {\log_2}^2 \pi_i
	\delta({\scriptstyle \sum_i} \pi_i - 1)
	\prod_{j\neq i} \pi_j^{\alpha_j - 1}
	d\vpi \\
& + \sum_{i\neq k} \int 
        \left(
          \pi_i^{\alpha_i} \log_2 \pi_i
        \right)
        \left(
          \pi_k^{\alpha_k} \log_2 \pi_k
        \right)
	\delta({\scriptstyle \sum_i} \pi_i - 1)
	\prod_{j\not \in \{i,k\}} \pi_j^{\alpha_j - 1}
	d\vpi
\\
&= 
	\sum_i  
        \frac{d^2}{d(\alpha_i+1)^2}
        \int 
         \pi_i^{\alpha_i+1}
	\delta({\scriptstyle \sum_i} \pi_i - 1)
	\prod_{j\neq i} \pi_j^{\alpha_j - 1}
	d\vpi \\
& +  
       \sum_{i\neq k} 
        \frac{d}{d\alpha_i}    \frac{d}{d\alpha_k}
        \int 
        \left(
          \pi_i^{\alpha_i}
        \right)
        \left(
          \pi_k^{\alpha_k}
        \right)
	\delta({\scriptstyle \sum_i} \pi_i - 1)
	\prod_{j\not \in \{i,k\}} \pi_j^{\alpha_j - 1}
	d\vpi
\end{align*}

Assuming $i \neq k$, these will be the cross terms.

\begin{align*}
      \qquad&
	\int 
	 (\pi_i \log_2 \pi_i)  (\pi_k \log_2 \pi_k)
	\delta({\scriptstyle \sum_i} \pi_i - 1)
	\prod_{j} \pi_j^{\alpha_j - 1}
	d\vpi
\\
&=       
 \deriv{}{\alpha_i} \deriv{}{\alpha_k}
	\int 
	 (\pi_i^{\alpha_i})  	 (\pi_k^{\alpha_k})  
	\delta({\scriptstyle \sum_i} \pi_i - 1)
	\prod_{j\not \in \{i,k\}} \pi_j^{\alpha_j - 1}
	d\vpi
\\
&=
	\deriv{}{\alpha_i}\deriv{}{\alpha_k}
	\left[
	    \frac{
		\Gamma(\alpha_i +1)
                \Gamma(\alpha_k +1)
	    }{
                \Gamma(A + 2)
	    }
	\right]
                \prod_{j\not \in \{i,k\}}
		\Gamma({\alpha_j})
\\
&=
	\deriv{}{\alpha_k}
	\left[
	 \frac{
		\alpha_i
                \Gamma(\alpha_k +1)
   	    }{
	     \Gamma(A+2)
	    }
           \psi_0(\alpha_i+ 1)  
         \right. \\ &\qquad\qquad\left.
           - \frac{
 		\alpha_i
                \Gamma(\alpha_k +1)
	    }{
	     \Gamma(A+2)
	    }
            \psi_0(A+ 2) 
	\right] \prod_{j \neq k} \Gamma({\alpha_j})
\\
&=
	\deriv{}{\alpha_k}
	\left[
	 \frac{
		\alpha_i
                \psi_0(\alpha_k +1)
   	    }{
	     \Gamma(A+2)
	    }
           \psi_0(\alpha_i+ 1)  
         \right. \\ &\qquad\qquad\left.
           - 
           \frac{
 		\alpha_i
                \Gamma(\alpha_k +1)
	    }{
	     \Gamma(A+2)
	    }
            \psi_0(A+ 2) 
	\right] \prod_{j \neq k} \Gamma({\alpha_j})
\\
&=
    \frac{ \alpha_i \alpha_k }{\Gamma(A + 2)}
\left[
  (\psi_0(\alpha_k+1)-\psi_0(A + 2))
\right. \\ &\qquad\left.   (\psi_0(\alpha_i+ 1) 
  -\psi_0(A + 2))-
  \psi_1(A + 2)
\right] \prod_j \Gamma({\alpha_j})
\\
&=
\frac{ \alpha_i \alpha_k }{(A+ 1)(A)}\left[
  \left(
	\psi_0(\alpha_k+1)-\psi_0(A + 2)
      \right) 
    \right. \\ &\qquad\left.
      \left(
	\psi_0(\alpha_i+1)-\psi_0(A + 2)
    \right)
	-\psi_1(A + 2)
    \right] 
    \frac{\prod_j \Gamma({\alpha_j})}{\Gamma(A)}
\end{align*}

\begin{align*}
&
\frac{d^2}{d(\alpha_i+1)^2}
        \int 
         \pi_i^{\alpha_i+1}
	\delta({\scriptstyle \sum_i} \pi_i - 1)
	\prod_{j\neq i} \pi_j^{\alpha_j - 1}
	d\vpi 
\\
&=
	\frac{d^2}{d(\alpha_i+1)^2}
	\left[
	    \frac{
		\Gamma(\alpha_i +2)
	    }{
                \Gamma(A + 2)
	    }
	\right]
                \prod_{j\neq i}
		\Gamma({\alpha_j})
\\
&=
	\frac{d^2}{dz^2}
	\left[
	    \frac{
		\Gamma(z +1)
	    }{
                \Gamma(z + c )
	    }
	\right]
                \prod_{j\neq i}
		\Gamma({\alpha_j}),~~~\left\{ c = A +2 -
                    (\alpha_i+1)\right\}
\\
&= 
\frac{\Gamma(1+z)}{\Gamma(c+z)}
\left[ (\psi_0( 1 + z) - 
\right. \\ &\qquad\left.
  \psi_0( c + z))^2 + \psi_1( 1 + z) - \psi_1(
  c + z)\right]
\prod_{j\neq i}
\Gamma({\alpha_j})
\\
&= 
\frac{
  z(z-1)
} {
  (c+z-1)(c+z-2)
}
\left[ (\psi_0( 1 + z) - \psi_0( c + z))^2 
\right. \\ &\qquad\qquad\qquad\left.
  + \psi_1( 1 + z) - \psi_1(c + z)\right]
\frac{
  \prod_{j}
  \Gamma({\alpha_j})
} {
  \Gamma(A)
}
\\
&= 
\frac{
  (\alpha_i+1)(\alpha_i)
} {
  (A + 1)(A )
}
\left[ (\psi_0( \alpha_i+2) - \psi_0( A + 2))^2 +
  \psi_1(\alpha_i+2) 
  \right. \\ &\qquad\qquad\qquad\left.
  - \psi_1(A + 2)\right]
\frac{
  \prod_{j}
  \Gamma({\alpha_j})
} {
  \Gamma(A)
}
\end{align*}

Summing over all terms and adding the cross and square terms, we
recover the desired expression for $\E[H^2(\vpi)|\valpha]$.
\end{proof}

\subsection{Prior entropy mean and variance under $\PY$}
We derive the prior entropy mean and variance of a $\PY$ distribution with fixed
parameters $\alpha$ and $d$, $\E_\vpi[H(\vpi)|d,\alpha]$ and $\var\vpi[H(\vpi)|d,\alpha]$. We first
prove our Proposition~\ref{prop:sizebiased:mean} (mentioned in
\cite{Pitman97}). This proposition establishes the identity $\E \left[ \sum_{i=1}^\infty f(\pi_i) \Big| \alpha \right] = \int_0^1 \frac{f(\tilde \pi_1) }{\tilde
  \pi_1} p(\tilde \pi_1 | \alpha) d\tilde \pi_1$  which will allow us to compute
expectations over $\PY$ using only the distribution of the first size
biased sample, $\tilde\pi_1$.
\begin{proof}[Proof of Proposition~\ref{prop:sizebiased:mean}]

First we validate \eqref{eq:sizebiasedexp}. Writing out the general form of the size-biased sample, 
\[
p(\tilde \pi_1 = x| \vpi)  = \sum_{i=1}^\infty \pi_i \delta(x-\pi_i),
\]
we see that 
\begin{align*}
\E_{\tilde \pi_1 } \left[ \frac{f(\tilde
        \pi_1)}{\tilde \pi_1} \right] &=   \int_0^1 \frac{f(x)}{x}
    p(\tilde \pi_1 = x)dx
\\
&= \int_0^1
  \E_\vpi\left[\frac{f(x)}{x} p(\tilde \pi_1 = x | \vpi)
  \right] dx
\\
&= \int_0^1
  \E_\vpi\left[ \sum_{i=1}^\infty \frac{f(x)}{x}\pi_i \delta(x-\pi_i) \right] dx
\\
&= \E_\vpi\left[\int_0^1
  \sum_{i=1}^\infty \frac{f(x)}{x}\pi_i \delta(x-\pi_i) dx\right]
\\
&= \E_\vpi\left[
  \sum_{i=1}^\infty \int_0^1 \frac{f(x)}{x}\pi_i \delta(x-\pi_i) dx\right]
\\
&= \E_\vpi \left[ \sum_{i=1}^\infty f(\pi_i) \right],
\end{align*}
where the interchange of sums and integrals is justified by Fubini's
theorem.

A similar method validates \eqref{eq:sizebiasedexpp}. We will need the
second size-biased sample in addition to the first. We begin with the
sum inside the expectation on the left--hand side of \eqref{eq:sizebiasedexpp},
\begin{align}
  &\sum_i \sum_{j \neq i} g(\pi_i,\pi_j)
  \\
    &=
    \frac{\sum_i \sum_{j \neq i} g(\pi_i,\pi_j)}
    {p(\tilde\pi_1 = \pi_i, \tilde\pi_2 = \pi_j)}
    p(\tilde\pi_1 = \pi_i, \tilde\pi_2 = \pi_j)
    \\
    &=
    \sum_i \sum_{j \neq i} \frac{g(\pi_i,\pi_j)}{\pi_i\pi_j} (1 - \pi_i)
    p(\tilde\pi_1 = \pi_i, \tilde\pi_2 = \pi_j)
    \\
    &=
    \E_{\tilde\pi_1, \tilde \pi_2}\left[
      \frac{g(\tilde \pi_1,\tilde \pi_2)}{\tilde \pi_1 \tilde \pi_2} (1 - \tilde\pi_1)
	\Big| \vpi
    \right]
\end{align}
where the joint distribution of size biased samples is given by,
\begin{align*}
    p(\tilde\pi_1 = \pi_i, \tilde\pi_2 = \pi_j)  &=
     p(\tilde\pi_1 = \pi_i) p(\tilde\pi_2 = \pi_j | \tilde\pi_1 =
    \pi_i) 
    \\
    &= \pi_i \cdot \frac{\pi_j}{1 - \pi_i}.
\end{align*}

\end{proof}

As this identity is defined for any additive functional $f$ of $\vpi$; we
can employ it to compute the first two moments of entropy. For PYP (and DP when $d=0$), the first size-biased sample is distributed according to:
\begin{equation}
\tilde \pi_1 \sim \Beta(1-d,\alpha+d)
\end{equation}

Proposition~\ref{prop:sizebiased:mean} gives the mean entropy directly. Taking $f(x) = -x
\log(x)$ we have, 
\begin{equation*}
  \E [H(\vpi)|d,\alpha] = -\E_\alpha [ \log(\pi_1)] = \psi_0(\alpha+1) - \psi_0(1-d),
    \end{equation*}

The same method may be used to obtain the prior variance, although the
computation is more involved. For the variance, we will need the
second size-biased sample in addition to the first. The second
size-biased sample is given by,

\begin{equation}
\tilde \pi_2 = (1-\tilde \pi_1) v_2, \quad v_2 \sim \Beta(1-d,\alpha+2d)
\end{equation}

We will compute the second moment explicitly, splitting $H(\vpi)^2$
into square and cross terms,

\begin{align}
\nonumber
    \E[\left( H(\vpi) \right)^2|d,\alpha]
    &= \E\left[\left.\left( - \sum_i \pi_i \log(\pi_i) \right)^2\right|d,\alpha\right]
    \\
   &= \E\left[\left.\sum_i \left(\pi_i \log(\pi_i)
       \right)^2\right|d,\alpha\right]
   \\
   &+ \E\left[\left.\sum_i \sum_{j \neq i} \pi_i \pi_j \log(\pi_i) \log(\pi_j)\right|d,\alpha\right]
    \label{eq:PY:prior:entropy:2nd}
\end{align}
The first term follows directly from \eqref{eq:sizebiasedexp},
\begin{align}
&\E\left[\left.\sum_i \left(\pi_i \log(\pi_i) \right)^2\right|d,\alpha\right]
     = \int_0^1 x (-\log(x))^2 p(x|d,\alpha) \dm{x}
    \nonumber \\
    &= 
	B^{-1}(1-d,\alpha+d) \int_0^1 x \log^2(x) x^{1-d} (1-x)^{\alpha+d-1} \dm{x}
    \nonumber \\
    &= 
	\frac{1-d}{\alpha+1} \left[(\psi_0(2-d) - \psi_0(2+\alpha))^2 + \psi_1(2-d) - \psi_1(2+\alpha)\right]
\end{align}

The second term of \eqref{eq:PY:prior:entropy:2nd}, requires the
first two size biased samples, and follows from
\eqref{eq:sizebiasedexpp} with $g(x,y) = \log(x)\log(y)$. For the PYP
prior, it is easier to integrate on $V_1$ and $V_2$, rather than the
size biased samples. The second term is then (note that we let
$\gamma = B^{-1}(1-d, \alpha+2d)$ and $\zeta = B^{-1}(1-d, \alpha+d)$),
\begin{align*}
    &\E\left[
	\E\left[
	    \log(\tilde\pi_1) \log(\tilde\pi_2) (1 - \pi_1)
	    | \vpi
	\right]
	| \alpha
    \right]
\\
    &=
    \E\left[
	\E\left[
	    \log(V_1) \log((1-V_1) V_2) (1 - V_1)
	    | \vpi
	\right]
	| \alpha
    \right]
    \\
    &=
	\zeta 	\gamma     \int_0^1 \int_0^1
	\log(v_1) \log((1-v_1) v_2) (1 - v_1)
	v_1^{1 - d}
	(1 - v_1)^{\alpha+d-1}
        \\
	&\qquad \qquad \qquad \qquad
        \times v_2^{1 - d}
	(1 - v_2)^{\alpha+2d-1}
    \dm{v_1} \dm{v_2}
    \\
    &=
    \zeta     \left[
    \int_0^1
	\log(v_1) \log(1-v_1) (1 - v_1) v_1^{1 - d} 
        (1 - v_1)^{\alpha+d-1}
        \dm{v_1}
      \right.
    \\
    &\quad+\left.
    \gamma     \int_0^1 
    \log(v_1) (1 - v_1) v_1^{1 - d} (1 - v_1)^{\alpha+d-1}
    \right.
    \\
    &\qquad\qquad\qquad\left.
      \times\int_0^1
	\log(v_2) v_2^{1 - d} (1 - v_2)^{\alpha+2d-1}
    \dm{v_1} \dm{v_2}
    \right]
    \\
    &=
    \frac{\alpha+d}{\alpha+1}\left[ (\psi_0(1-d)-\psi_0(2+\alpha))^2 - \psi_1(2+\alpha)\right]
    \label{eq:PY:prior:entropy:var:derivation}
\end{align*}

Finally combining the terms, the variance of the entropy under PYP prior is
\begin{align}
  &\var[H(\vpi)|d,\alpha] =
\\
&\frac{1-d}{\alpha+1} \left[(\psi_0(2-d) - \psi_0(2+\alpha))^2 + \psi_1(2-d) - \psi_1(2+\alpha)\right]
\nonumber\\
&\quad\quad
+\frac{\alpha+d}{\alpha+1}\left[ (\psi_0(1-d)-\psi_0(2+\alpha))^2 - \psi_1(2+\alpha)\right]
\nonumber\\
&\qquad
-(\psi_0(1+\alpha) - \psi_0(1-d))^2
\nonumber\\
&=\frac{\alpha+d}{(\alpha+1)^2 (1-d)} + \frac{1-d}{\alpha+1}\psi_1(2-d) - \psi_1(2+\alpha)
\end{align}

We note that the expectations over the finite Dirichlet may also be
derived using this formula by letting the $\tilde \vpi$ be the first
size-biased sample of a finite Dirichlet on $\Delta_{\alphabetSize}$. 

\subsection{Posterior Moments of PYP} 
\label{sec:appendix:posterior}
First, we discuss the form of the PYP posterior, and introduce
independence properties that will be important in our derivation of the
mean. We recall that the PYP posterior, $\vpi_{\mathrm{post}}$, of
\eqref{eq:postPY} has three stochastically independent components: Bernoulli
$p_\ast$, PY $\vpi$, and Dirichlet $\vp$.  
  
\textbf{Component expectations:} 
From the above derivations for expectations under the PYP and Dirichlet
distributions as well as the Beta integral identities (see e.g., \cite{Archer2012b}),  we find expressions for  
$\E_{\vp}\left[ H(\vp) | d,\alpha \right]$,
$E_{\vpi}\left[ H(\vpi) |d,\alpha\right]$, and 
$\E_{p_\ast}\left[ H(p_\ast)\right]$. 
\begin{align*}
  \E [H(\vpi)|d,\alpha] &= \psi_0(\alpha+1) - \psi_0(1-d)
\\
  \E_{p_\ast}[ H(p_\ast) ] 
    &= \psi_0( \alpha + N + 1) 
      - \frac{\alpha +Kd}{\alpha+N} \psi_0(\alpha+Kd +1) 
\\ &\qquad\quad
      - \frac{N-Kd}{\alpha+N} \psi_0(N-Kd+1)
\\
\E_{\vp}[H(\vp)|d,\alpha] 
    &= \psi_0(N-Kd+1) - \sum_{i=1}^K  \frac{n_i-d}{N-Kd} \psi_0(n_i-d+1)
\end{align*}
where by a slight abuse of notation we define the entropy of $p_\ast$
as $H(p_\ast) = -(1-p_\ast)\log(1-p_\ast) - p_\ast \log p_\ast$. We
use these expectations below in our computation of the final posterior
integral.

\textbf{Derivation of posterior mean:}
We now derive the analytic form of the posterior mean, \eqref{eq:posterior:mean}.
\begin{align*}
  &\E[ H(\pi_{\mathrm{post}})| d,\alpha] = \E\left[ - \sum_{i=1}^K p_i
  \log{p_i} - p_\ast \sum_{i=1}^\infty \pi_i \log{ p_\ast \pi_i} \Big|
  d,\alpha \right]\nonumber
\\\nonumber
&= \E\left[  - (1-p_\ast)\sum_{i=1}^K \frac{p_i}{1-p_\ast}
\log{\left(\frac{p_i}{1-p_\ast}\right)} 
\right.\\&\quad \left. 
- (1-p_\ast)\log(1-p_\ast) - p_\ast
\sum_{i=1}^\infty \pi_i \log{\pi_i} - p_\ast \log{ p_\ast} \Big|
  d,\alpha \right]
\\\nonumber
&= \E\left[  - (1-p_\ast)\sum_{i=1}^K \frac{p_i}{1-p_\ast}
\log{\left(\frac{p_i}{1-p_\ast}\right)} 
\right.\\&\qquad\qquad\qquad \left. 
- p_\ast
\sum_{i=1}^\infty \pi_i \log{\pi_i}    + H(p_\ast) \Big|
  d,\alpha \right]
\\\nonumber
&= \E\left[ \E\left[ - (1-p_\ast)\sum_{i=1}^K \frac{p_i}{1-p_\ast}
\log{\left(\frac{p_i}{1-p_\ast}\right)} 
\right.\right.\\&\qquad\qquad\qquad \left. \left.
- p_\ast \sum_{i=1}^\infty \pi_i \log{\pi_i}    + H(p_\ast)  ~\Big| ~ p_\ast
\right] \Big|
  d,\alpha \right]
\\\nonumber
&= \E\left[ \E\left[(1-p_\ast) H(\vp)+ p_\ast H(\vpi) + H(p_\ast)  ~\Big| ~ p_\ast
  \right] \Big|
  d,\alpha \right]\nonumber
\\
&= \E_{p_\ast}\left[ (1-p_\ast) \E_{\vp}\left[ H(\vp) | d,\alpha \right] +  p_\ast\E_{\vpi}\left[ H(\vpi)
  |d,\alpha\right] + H(p_\ast)    \right]\nonumber
\end{align*}

using the formulae for $\E_{\vp}\left[ H(\vp) | d,\alpha \right]$, $\E_{\vpi}\left[ H(\vpi)
  |d,\alpha\right]$, and $\E_{p_{\ast}}[H(p_{\ast})]$ and rearranging terms, we obtain \eqref{eq:posterior:mean}, 
\begin{align*}
&\E[ H(\pi_{\mathrm{post}})| d,\alpha] 
    = \frac{A}{\alpha+N}  \E_{\vp}[ H(\vp) ] 
    \\  &\qquad\qquad\qquad 
    + \frac{\alpha + Kd}{\alpha+N} \E_{\vpi}[ H(\vpi) ]  
    + \E_{p_\ast}[ H(p_\ast) ] \nonumber
\\ \nonumber
    &= \frac{A}{\alpha+N} \left[\psi_0(A+1) - \sum_{i=1}^K
      \frac{\alpha_i}{A} \psi_0(\alpha_i+1)\right] 
    \\&\quad
    + \frac{\alpha + Kd}{\alpha+N}
    \left[ \psi_0(\alpha+Kd+1) 
      - \psi_0(1-d)\right] +
    \\\nonumber
    &\quad \psi_0( \alpha + N + 1) - \frac{\alpha +Kd}{\alpha+N}
    \psi_0(\alpha+Kd +1) - \frac{A}{\alpha+ N} \psi_0(A+1)
\\\nonumber
&= \psi_0(\alpha+N+1) - \frac{\alpha+Kd}{\alpha+N}\psi_0(1-d) - 
    \\&\qquad \qquad \qquad \qquad \qquad
\frac{A}{\alpha+N} \left[\sum_{i=1}^K
    \frac{\alpha_i}{A} \psi_0(\alpha_i+1)\right] \nonumber
\\
&= \psi_0(\alpha+N+1) - \frac{\alpha+Kd}{\alpha+N}\psi_0(1-d) -
    \\&\qquad \qquad \qquad \qquad \qquad
    \frac{1}{\alpha+N} \left[
\sum_{i=1}^K
    (n_i - d) \psi_0(n_i - d+1)\right] \nonumber
\end{align*}

\textbf{Derivation of posterior variance:}
We continue the notation from the subsection above. In order to
exploit the independence properties of $\pi_{\mathrm{post}}$ we first apply
the law of total variance to obtain \eqref{eq:PY:posterior:variance},

\begin{align}
 \label{eq:PY:posterior:variance}
 \var[ H(\pi_{\mathrm{post}})| d,\alpha] &= \var_{p_\ast}\left[
    \E_{\vpi, \vp}[H(\pi_{\mathrm{post}}) ]  \Big| d,\alpha\right] \nonumber
      \\&\qquad 
  + \E_{p_\ast}\left[
    \var_{\vpi, \vp}[H(\pi_{\mathrm{post}}) ] \Big| d,\alpha\right]
\end{align}

We now seek expressions for each term in
\eqref{eq:PY:posterior:variance} in terms of the expectations already
derived. 

\textit{Step 1:} For the right-hand term of
\eqref{eq:PY:posterior:variance}, we use the independence properties of
$\pi_{\mathrm{post}}$ to express the variance in terms of PYP,
Dirichlet, and Beta variances, 
\begin{align}
\label{eq:PY:posterior:variance:term1}
&\E_{p_\ast}\left[\var_{\vpi, \vp}[H(\pi_{\mathrm{post}}) | p_\ast] \Big|
  d,\alpha\right] \\&= \E_{p_\ast}\left[  (1-p_\ast)^2\var_{\vp}[H(\vp)] 
  + p_\ast^2 \var_{\vpi}[H(\vpi)]  \Big|d,\alpha\right] \nonumber
\\
&= \frac{(N-Kd)(N-Kd+1)}{(\alpha+N)(\alpha+N+1)}\var_{\vp}[H(\vp)]\nonumber
\\
&\qquad \qquad + \frac{(\alpha+Kd)(\alpha+Kd+1)}{(\alpha+N)(\alpha+N+1)}
\var_{\vpi}[H(\vpi)]  
\end{align}
\textit{Step 2:} In the left-hand term of
\eqref{eq:PY:posterior:variance}  the variance is with respect to the
$\Beta$ distribution, while the inner expectation is precisely the
posterior mean we derived above. Expanding, we obtain,
\begin{align}
\label{eq:PY:posterior:variance:term2}
\nonumber
&\var_{p_\ast}\left[  \E_{\vpi, \vp}[H(\pi_{\mathrm{post}}) | p_\ast]  \Big| d,\alpha\right]\\&= \var_{p_\ast}\left[ (1-p_\ast) \E_{\vp}\left[ H(\vp) \right] +  p_\ast\E_{\vpi}\left[ H(\vpi)
    | p_\ast \right] + H(p_\ast) \Big|
  d,\alpha \right] 
\end{align}
To evaluate this integral, we introduce some new notation,
\begin{align*}
\mathbf{A} &= \E_{\vp}\left[ H(\vp)\right]
\\
\mathbf{B} &= \E_{\vpi}\left[ H(\vpi)\right] 
\\
\Omega(p_\ast) &=  (1-p_\ast) \E_{\vp}\left[ H(\vp) \right] +  p_\ast\E_{\vpi}\left[
  H(\vpi)\right] + H(p_\ast)
\\
&= (1-p_\ast) \mathbf{A} +  p_\ast\mathbf{B} + H(p_\ast)
\end{align*}
so that 
\begin{align}\nonumber
\Omega^2(p_\ast)&= 2p_\ast H(p_\ast)[\mathbf{B}-\mathbf{A}]  +
2\mathbf{A}H(p_\ast) + h^2(p_\ast) \\ &+ p_\ast^2[\mathbf{B}^2 - 2\mathbf{A}\mathbf{B}]
+ 2p_\ast\mathbf{A}\mathbf{B} + (1-p_\ast)^2\mathbf{A}^2\label{eq:Kpstarsquared}
\end{align}
and we note that 
\begin{equation}
  \var_{p_\ast}\left[  \E_{\vpi, \vp}[H(\pi_{\mathrm{post}}) | p_\ast] \Big| d,\alpha\right] = \E_{p_\ast}[\Omega^2(p_\ast)]  - \E_{p_\ast}[\Omega(p_\ast)]^2
\end{equation}
The components composing $\E_{p_\ast}[\Omega(p_\ast)]$, as well as each term of
\eqref{eq:Kpstarsquared} can be found in~\cite{Archer2012b}. Although less
elegant than the posterior mean, the
expressions derived above permit us to compute \eqref{eq:PY:posterior:variance}  numerically from its component expectations, without
sampling.

\section{Proof of Proposition \ref{finiteposterior}}\label{sec:appendix:finiteposterior}
In this Appendix we give a proof for Proposition
\ref{finiteposterior}. 
\begin{proof}
  PYM is given by 
\begin{equation*}
\hat H_{\PYM} = \frac{1}{p(\vx)} \int_0^\infty \int_0^1
    H_{(d,\alpha)} p(\vx | d,\alpha) p(d,\alpha) \dm{\alpha} \dm{d}
\end{equation*}
where we have written $H_{(d,\alpha)} = \E[H | d,\alpha,  \vx]$. Note
that $p(\vx|d,\alpha)$ is the evidence, given by \eqref{eq:evidence}.
We will assume $p(d,\alpha) = 1$ for all $\alpha$ and $d$ to show conditions under which
$H_{(d,\alpha)}$ is integrable for any prior. Using the identity
$\frac{\Gamma(x+n)}{\Gamma(x)} = \prod_{i=1}^n (x+i-1)$ and the log
convexity of the Gamma function we have, 

\begin{align*}
  p(\vx|d,\alpha) &\leq 	
  \prod_{i=1}^K \frac{\Gamma(n_i - d)}{\Gamma(1-d)}
\frac{
    \Gamma(\alpha+K)
  }{
    \Gamma(\alpha + N)
  }
\\
&\leq \frac{\Gamma(n_i - d)}{\Gamma(1-d)} \frac{1}{\alpha^{N-K}}
\end{align*}

Since $d\in[0,1)$, we have from the properties of the
digamma function,
\begin{equation*}
  \psi_0(1-d) = \psi_0(2-d) - \frac{1}{1-d} \geq \psi_0(1) -
  \frac{1}{1-d} = -\gamma - \frac{1}{1-d},
\end{equation*}
and thus the upper bound, 
\begin{align}\label{eq:entropyupper}
  H_{(d,\alpha)} &\leq
  \psi_0(\alpha+N+1) + \frac{\alpha+Kd}{\alpha+N}\left(\gamma +
    \frac{1}{1-d}\right) \\&-
  \frac{1}{\alpha+N} \left[\sum_{i=1}^K (n_i - d) \psi_0(n_i -
    d+1)\right].
\end{align}

Although second term is unbounded in
$d$ notice that $\frac{\Gamma(n_i - d)}{\Gamma(1-d)} =
\prod_{i=1}^{n_i} (i - d)$; thus, so long as $N-K\geq 1$,
$H_{(\alpha,d)} p(\vx | d,\alpha)$ is integrable in $d$. 

For the integral over alpha, it suffices to choose $\alpha_0\gg N$
and consider the tail, $\int_{\alpha_0}^\infty H_{(d,\alpha)} p(\vx |
d,\alpha) p(d,\alpha) \dm{\alpha}$. From \eqref{eq:entropyupper} and
the asymptotic expansion $\psi(x) = \log(x)  - \frac{1}{2x} -
\frac{1}{12x^2} + O(\frac{1}{x^4})$ as $x\to\infty$ we see
that in the limit of $\alpha\gg N$, 
\begin{equation*}
  H_{(d,\alpha)} \leq
    \log(\alpha+N+2) + \frac{c}{\alpha},
\end{equation*}
where $c$ is a constant depending on $K$,  $N$, and $d$. Thus, we have
\begin{align*}
\int_{\alpha_0}^\infty 
    &H_{(d,\alpha)} p(\vx | d,\alpha) p(d,\alpha) \dm{\alpha} \\&\leq
    \frac{\prod_{i=1}^K\Gamma(n_i-d)}{\Gamma(1-d)}\int_{\alpha_0}^\infty
    \left(\log(\alpha+N+2) + \frac{c}{\alpha}\right) \frac{1}{\alpha^{N-K}}\dm{\alpha} 
\end{align*}
and so $H_{(d,\alpha)}$ is integrable in $\alpha$ so long as $N-K\geq 2$.
\end{proof}

\section{Proofs of Consistency Results}
\begin{proof}[proof of Theorem \ref{fixedHconv}]
We have,  
  \begin{align*}
\lim_{N\to\infty} &\E[H | \alpha, d,
  \vx_N]\\ &= \lim_{N\to\infty} \left[ \psi_0(\alpha+N+1) -
    \frac{\alpha+K_Nd}{\alpha+N}\psi_0(1-d) - 
    \right. \\ &\qquad\qquad\left.
    \frac{1}{\alpha+N}
    \left[\sum_{i=1}^{K_N}  (n_i - d) \psi_0(n_i - d+1)\right] \right]
\\
&= \lim_{N\to\infty} \left[ \psi_0(\alpha+N+1) - \sum_{i=1}^{K_N}
    \frac{n_i}{N} \psi_0(n_i - d+1) \right]
\\
&= -\lim_{N\to\infty} \sum_{i=1}^{K_N}
    \frac{n_i}{N} \left[\psi_0(n_i - d+1) - \psi_0(\alpha+N+1) \right]
  \end{align*}

although we have made no assumptions about the tail behavior of
$\vpi$, so long as $\pi_k>0$,  $\E[n_k] = \E[\sum_{i=1}^\infty \ind{x_i = k}] =
\sum_{i=1}^\infty P\{ x_i = k\} =\lim_{N\to\infty} N\pi_k \to \infty$,
and we may apply the asymptotic expansion $\psi(x) = \log(x)  - \frac{1}{2x} -
\frac{1}{12x^2} + O(\frac{1}{x^4})$ as $x\to\infty$ to find, 

\begin{equation*}
\lim_{N\to\infty}  \E[H | \alpha, d,  \vx_N]  = H_{\mathrm{plugin}}
\end{equation*}
\end{proof}

We now turn to the proof of consistency for PYM. Although consistency
is an intuitively plausible property for PYM, due to the form of the
estimator our proof involves a rather detailed technical argument.
Because of this, we break the proof of Theorem \ref{consistencyproof}
into two parts. First, we prove a supporting Lemma.

\begin{lemma}
\label{lem:convergence:head}
    If the data $\vx_N$ have at least two coincidences, and are sampled
    from a distribution such that, for some constant $c>0$, $K_N =
    o(N^{1-1/c})$ in probability, the following sequence of integrals converge.
\begin{align*}
    \int_0^{K_N+c} \int_0^1
    \E[H|\alpha, d, \vx_N] 
    \frac{p(\vx_N | \alpha, d) p(\alpha, d)}{p(\vx_N)}
    \mathrm{d}\alpha \mathrm{d}d
    \xrightarrow{P} \E[\Hplug| \vx_N]
\end{align*}
where $c>0$ is an arbitrary constant. 
\end{lemma}
\begin{proof}
  
Notice first that $E[H|\alpha,d,\vx_N]$ is monotonically increasing in
$\alpha$, and so 

\begin{align*}
  \int_{\alpha=0}^{K_N+c}&\int_{d=0}^1 \E[H|\alpha,d,\vx_N]
  \frac{p(\vx_N|\alpha,d)}{p(\vx_N)} \dm{\alpha}\dm{d} \\&\leq
  \int_{\alpha=0}^{K_N+c}\int_{d=0}^1
  \E[H|K_N+c,d,\vx_N]\frac{p(\vx_N|\alpha,d)}{p(\vx_N)} \dm{\alpha}\dm{d}.
\end{align*}
As a result we have that, 

\begin{align}\label{eq:entropybound_head}
&\E[H | K_N+c, d, \vx_N] = \psi_0(K_N+c+N+1) \\&\qquad\qquad- \frac{
  (1+d)K_N+c}{K_N+N+c}\psi_0(1-d) \nonumber
    \\&\qquad\qquad- \frac{1}{K_N+c+N} \left(\sum_{i=1}^{K_N} (n_i - d) \psi_0(n_i -
      d+1)\right)\nonumber
\end{align}
As a consequence of Proposition \ref{finiteposterior}, $\int_{d=0}^{1}
(1+d) \psi(1-d) \frac{p(\vx|\alpha,d)}{p(\vx_N)} \dm{d}<\infty$, and
so the second term is bounded and controlled by $K_N/N$. We let 
\[
A(d,N) =- \frac{ (1+d)K_N+c}{K_N+N+c}\psi_0(1-d)
\]
and, since $\lim_{N\to \infty} \int_{d=0}^1 A(d,N) \frac{p(\vx|\alpha,d)}{p(\vx_N)}\dm{d}
= 0$, we focus on the
remaining terms of  \eqref{eq:entropybound_head}. We also let
$B(\vn) = \sum_{i=1}^{K_N}\left(\frac{n_i-1}{N}
    \log\left(\frac{n_i}{N}\right)\right)$, and note that $\lim_{N\to
    \infty} B = \Hplug$.  We find that, 
\begin{align*}
  \E[H | &K_N+c, d, \vx_N] \\ &\leq 
\log(N+K_N+c+1)
   +A(d,N) 
   \\ &
   - \sum_{i=1}^{K_N}\left(\frac{n_i-1}{K_N+N +c}
     \log(n_i)\right)
\\
&=\log(N+K_N+c+1) + A(d,N) -
\\&\frac{N}{K_N+N +c}
\left[\sum_{i=1}^{K_N}\left(\frac{n_i-1}{N}
    \log\left(\frac{n_i}{N}\right)\right) +
  \frac{N-K_N}{N}\log(N)\right]
\\
&=\log\left(1 + \frac{K_N+c+1}{N}\right) + A(d,N) 
\\ &\quad + \log(N) \left[ \frac{2K_N+c}{N+K_N+c} \right]+ \frac{N}{K_N+N +c}
B
\\
&=\log\left(1 + \frac{K_N+c+1}{N}\right) + A(d,N)  
\\ &\quad+ \frac{1}{1+(K_N+c)/N}\frac{2K_N+c}{N^{1-1/C}}\frac{\log(N)}{N^{1/C}}+ \frac{N}{K_N+N +c}
B
\\
&\to  \Hplug+ o(1)
\end{align*}

As a result, 
\begin{align*}
  \int_{\alpha=0}^{K_N+c}\int_{d=0}^1 &\E[H|\alpha,d,\vx_N]
  \frac{p(\vx_N|\alpha,d)}{p(\vx_N)}\dm{d} \dm{\alpha} \\&\leq
  \left[ \Hplug
    \int_{\alpha=0}^{K_N+c} \int_{d=0}^1 \frac{p(\vx_N|\alpha,d)}{p(\vx_N)}
    \dm{d}\dm{\alpha} + o(1)\right]
\\
&\to \Hplug
\end{align*}

For the lower bound, we let $H_{(\alpha,d,N)} =
\E[H|\alpha,d,\vx_N]\indnobr{[0,K_N+c]}(\alpha)$. Notice that $\exp(-H_{(\alpha,d,N)}) \leq 1$, so by dominated convergence
$\lim_{N\to\infty} \E[\exp(-H_{(\alpha,d,N)})] = \exp(-\Hplug)$ by
Proposition \ref{finiteposterior}. And so by
Jensen's inequality,

\begin{align*}
  \exp(-\lim_{N\to\infty}\E[H_{(\alpha,d,N)}]) &\leq
  \lim_{N\to\infty}\E[\exp(-H_{(\alpha,d,N)})] = \exp(-\Hplug)
\\
\implies & \lim_{N\to\infty} \E[H_{(\alpha,d,N)}]  \geq \Hplug,
\end{align*}

and the lemma follows. 
\end{proof}

~\\
We now turn to the proof of our primary consistency result. \\

\begin{proof}[proof of Theorem \ref{consistencyproof}]

\begin{align*}
    \iint &\E[H|\alpha, d, \vx_N] 
    \frac{p(\vx_N | \alpha, d) p(\alpha, d)}{p(\vx_N)}
    \mathrm{d}\alpha \mathrm{d}d
    \\
    &= 
    \int_0^{\alpha_0} \int_0^1
    \E[H|\alpha, d, \vx_N] 
    \frac{p(\vx_N | \alpha, d) p(\alpha, d)}{p(\vx_N)}
    \mathrm{d}\alpha \mathrm{d}d
    \\
    &\quad+
    \int_{\alpha_0}^\infty \int_0^1
    \E[H|\alpha, d, \vx_N] 
    \frac{p(\vx_N | \alpha, d) p(\alpha, d)}{p(\vx_N)}
    \mathrm{d}\alpha \mathrm{d}d
\end{align*}
If we let $\alpha_0 = K_N+1$,
by Lemma \ref{lem:convergence:head},
\begin{align*}
    \int_0^{\alpha_0} \int_0^1
    \E[H|\alpha, d, \vx_N] 
    \frac{p(\vx_N | \alpha, d) p(\alpha, d)}{p(\vx_N)}
    \mathrm{d}\alpha \mathrm{d}d
    \to \E[H_\mathrm{plugin}| \vx_N].
\end{align*}
Therefore, it remains to show that 
\begin{align*}
    \int_{\alpha_0}^\infty \int_0^1
    \E[H|\alpha, d, \vx_N] 
    \frac{p(\vx_N | \alpha, d) p(\alpha, d)}{p(\vx_N)}
    \mathrm{d}\alpha \mathrm{d}d 
    \to 0
\end{align*}
For finite support distributions where $K_N \to K < \infty$, this is trivial.
Hence, we only consider infinite support distributions where $K_N \to \infty$.
In this case, there exists $N_0$ such that for all $N \geq N_0$, $p([0, K_N+1], [0, 1)) \neq 0$.

Since $p(\alpha, d)$ has a decaying tail as $\alpha \to \infty$, $\exists N_0 \forall N \geq N_0$, $p(K_N+1, d) \leq 1$, thus, it is sufficient demonstrate convergence under an improper prior $p(\alpha, d) = 1$.

Using,
\begin{align*}
    \E[H|\alpha, d, \vx_N]  \leq \psi_0(N + \alpha + 1) \leq N + \alpha
\end{align*}
we bound
\begin{align*}
    \int_{\alpha_0}^\infty \int_0^1
    \E[H|\alpha, d, \vx_N] 
    &\frac{p(\vx_N | \alpha, d)}{p(\vx_N)}
    \mathrm{d}\alpha \mathrm{d}d 
    \\&\leq
    \frac{
    \int_{\alpha_0}^\infty \int_0^1
    (N + \alpha - 1) p(\vx_N | \alpha, d)
    \mathrm{d}\alpha \mathrm{d}d 
    }{p(\vx_N)}
    \\&\qquad+
    \frac{
    \int_{\alpha_0}^\infty \int_0^1
    p(\vx_N | \alpha, d)
    \mathrm{d}\alpha \mathrm{d}d 
    }{p(\vx_N)}
\end{align*}
We focus upon the first term on the RHS since its boundedness implies that of the smaller second term.
Recall, that $p(\vx) = \int_{\alpha=0}^{\infty} \int_{d=0}^1 p(\vx|\alpha,d) \dm d \dm \alpha$.
We seek an upper bound for the numerator and a lower bound for $p(\vx_N)$.

\textit{Upper Bound:} First we integrate over $d$ to find the upper
bound of the numerator. (For the following display only we let $\gamma(d) = \prod_{i=1}^{K_N} \Gamma(n_i - d)$).
\begin{align*}
    &\int_{\alpha_0}^\infty \int_0^1
    (N + \alpha - 1)
    p(\vx_N | \alpha, d)
    \mathrm{d}\alpha \mathrm{d}d 
    \\&=
    \int_{\alpha_0}^\infty \int_{d=0}^1	
    \frac{
      \left(
        \prod_{l=1}^{K_N-1} (\alpha + ld)
      \right)
      \gamma(d)
      \Gamma(1 + \alpha)
	(N + \alpha - 1)
    }{
      \Gamma(1 - d)^{K_N}
      \Gamma(\alpha + N)
    } \dm d \dm \alpha 
    \\
    &\leq
    \int_{d=0}^1
    \frac{
      \gamma(d)
    }{
      \Gamma(1 - d)^{K_N}
    }
    \dm{d}
    \int_{\alpha_0}^\infty
    \frac{
      \Gamma(\alpha  + K_N)
	(N + \alpha - 1)
    }{
      \Gamma(\alpha + N)
    } \dm \alpha 
\end{align*}
Fortunately, the first integral on $d$ will cancel with a term from the lower bound of $p(\vx_N)$.
Using\footnote{Note that in the argument for the inequalities we use
  $K$ rather than $K_N$ for clarity of notation.},
$
\frac{ (N + \alpha - 1)\Gamma(\alpha + K_N)}{\Gamma(\alpha+N)} 
= \frac{\Beta(\alpha+K_N,N-K-1)}{\Gamma(N-K-1)}
$,
\begin{align*}
  &\int_{\alpha_0}^{\infty}
  \frac{
    (N + \alpha - 1)
  \Gamma(\alpha+K)}{\Gamma(\alpha+N)} \dm \alpha \\&=
  \frac{1}{\Gamma(N-K-1)}
  \int_{\alpha_0}^\infty \Beta(\alpha+K,N-K-1) \dm \alpha
\\
&=
  \frac{1}{\Gamma(N-K-1)}
  \int_{\alpha_0}^\infty \int_{0}^1 t^{\alpha+K-1}(1-t)^{N-K-2} \dm t \dm \alpha
\\
&=
  \frac{1}{\Gamma(N-K-1)}
  \int_{t=0}^1 \frac{t^{\alpha_0+K-1}}{\log(\frac{1}{t})}(1-t)^{N-K-2} \dm{t}
\\
&\leq
  \frac{1}{\Gamma(N-K-1)}
\int_{t=0}^1 \frac{t^{\alpha_0+K-1}}{(1-t)}(1-t)^{N-K-2} \dm t
\\
&=
  \frac{1}{\Gamma(N-K-1)}
\Beta(\alpha_0 + K, N-K-2)
\\
&=
\frac{1}{\Gamma(N-K-1)}
\frac{ \Gamma(\alpha_0 + K) \Gamma(N-K-2)}{\Gamma(N + \alpha_0 - 2)}
\\
&=
\frac{ \Gamma(\alpha_0 + K) }{\Gamma(N + \alpha_0 - 2)(N-K-2)}
\end{align*}

\textit{Lower Bound:} 
Again, we first integrate $d$,
\begin{align*}
  &\int_{\alpha=0}^{\infty} \int_{d=0}^1 p(\vx|\alpha,d) \dm d \dm \alpha
  \\&=
    \int_{\alpha=0}^{\infty} \int_{d=0}^1
    \frac{
      \left(
        \prod_{l=1}^{K-1} (\alpha + ld)
      \right)
      \left(
        \prod_{i=1}^{K} \Gamma(n_i - d)
      \right)
      \Gamma(1 + \alpha)
    }{
      \Gamma(1 - d)^{K}
      \Gamma(\alpha + N)
    } \dm d \dm \alpha
\\
&=
    \int_{d=0}^1
    \frac{
      \left(
        \prod_{i=1}^{K} \Gamma(n_i - d)
      \right)
    }{
      \Gamma(1 - d)^{K}
    }
    \dm{d}
    \int_{\alpha=0}^{\infty}
    \frac{
      \alpha^{K-1}\Gamma(1 + \alpha)
    }{
      \Gamma(\alpha + N)
    }\dm \alpha
\end{align*}

So, since $\frac{\Gamma(1+\alpha)}{\Gamma(\alpha+N)} =
\frac{\Beta(1+\alpha,N-1)}{\Gamma(N-1)}$, then 

\begin{align*}
\Gamma(N-1)
  \int_{\alpha=0}^\infty &\frac{\alpha^{K-1}\Gamma(1+\alpha)}{\Gamma(\alpha+N)} \dm 
  \alpha \\&\geq
  \int_{\alpha=0}^\infty
  \alpha^{K-1} \Beta(1+\alpha, N-1)
  \dm \alpha
  \\
  &=
  \int_{\alpha=0}^\infty
  \alpha^{K-1} \int_{t=0}^1 t^{\alpha} (1-t)^{N-2} \dm t \dm \alpha
  \\
  &=
  \int_{t=0}^1 
  (1-t)^{N-2}
  \int_{\alpha=0}^\infty
  \alpha^{K-1} t^{\alpha} \dm \alpha \dm t 
  \\
  &=
 \Gamma(K) 
  \int_{t=0}^1 
  (1-t)^{N-2}
  \log\left(\frac{1}{t}\right)^{-K} \dm t 
  \\
  &\geq
  \Gamma(K) 
  \int_{t=0}^1 
  (1-t)^{N-K-2}
  t^{K} \dm t 
  \\
  &= 
 \Gamma(K) 
  \Beta(N-K-1, K+1)
\end{align*}
where we've used the fact that $\log(\frac{1}{t})\inv \geq \frac{t}{1-t}$.
Finally, we obtain the bound, 
\begin{align*}
    \int_{\alpha=0}^\infty \frac{\alpha^{K_N-1}\Gamma(1+\alpha)}{\Gamma(\alpha+N)} \dm \alpha   
  &\geq
  \frac{
    \Gamma(K)\Gamma(N-K-1) \Gamma(K+1)
  }{
    \Gamma(N-1)\Gamma(N)
  }.
\end{align*}
Now, we apply the upper and lower bounds to bound PYM.  We have, 
\begin{align*}
    &\frac{
    \int_{\alpha_0}^\infty \int_0^1
    (N + \alpha - 1) p(\vx_N | \alpha, d)
    \mathrm{d}\alpha \mathrm{d}d 
    }{p(\vx_N)}
    \\&\leq
\frac{
    \Gamma(\alpha_0 +K_N)
}{
    (N-K_N-2)\Gamma(N+\alpha_0 - 2)
}
\frac{
    \Gamma(N-1)\Gamma(N)
}{
    \Gamma(K_N)\Gamma(N-K_N-1) \Gamma(K_N+1)
}
\\
&= 
\frac{
1
}{
  (N-K_N-2)
}
\frac{
  \Gamma(\alpha_0 +K_N) 
}{
  \Gamma(K_N) 
}
\frac{
  \Gamma(N-1)
}{
  \Gamma(N+\alpha_0 - 2) 
}
\\&\qquad \qquad \qquad\qquad \qquad \qquad \qquad\times
\frac{ 
  \Gamma(N)
}{
  \Gamma(N-K_N-1) \Gamma(K_N+1)
} 
\\
&\to
\frac{
N
}{
  (N-K_N-2)
}
\left(\frac{
 K_N
}{
  N
}\right)^{\alpha_0}
\frac{ 
  N^{N-1/2}
}{
  (N-K_N-1)^{N-K_N-3/2} (K_N+1)^{K_N+1/2}
} 
\\
&=
\frac{
    N^2
}{
    (K_N+1)^{1/2}  (N-K_N-2)
}
\left(\frac{
    K_N
}{
    N
}\right)^{\alpha_0}
\left(
\frac{ 
  N
}{
  N-K_N-1
} \right)^{N-3/2}
\\&\qquad \qquad \qquad\qquad \qquad \qquad \qquad\times
\left(
\frac{ 
N-K_N-1
}{
K_N+1
} \right)^{K_N}
\\
&\to
\frac{
N
}{
(K_N+1)^{1/2} 
}
\left(\frac{
    K_N
}{
    N
}\right)^{\alpha_0}
\left(
\frac{ 
    N
}{
    K_N
} \right)^{K_N}
\end{align*}
Where we have applied the asymptotic expansion for the Beta function, 
\[
\Beta(x,y) \sim \sqrt {2\pi } \frac{{x^{x - \frac{1}{2}} y^{y - \frac{1}{2}} }}{{\left( {x + y} \right)^{x + y - \frac{1}{2}} }},
\]a consequence of Stirling's formula. Finally, we take $\alpha_0 = K_N
+ (C+1)/2$ so that the limit becomes,
\begin{align*}
&\to \frac{
N
}{
K_N^{1/2}
}
\left(\frac{
 K_N
}{
  N
}\right)^{(C+1)/2}
\\
&=
\frac{
    K_N^{C/2}
}{
    N^{C/2-1/2}
}
\end{align*}
which tends to $0$ with increasing $N$ since, by assumption, $K_N=o(N^{1-1/C})$. 
\end{proof}

\makeatletter{}

\section{Results on Unimodality of Evidence}\label{sec:appendix:unimodal}
\begin{theorem}[Unimodal evidence on $d$] \label{thm:unimodal:evidence:d}
    The evidence $p(\vx|d,\alpha)$ given by \eqref{eq:evidence} has only one local maximum (unimodal) for a fixed $\alpha>0$.
\end{theorem}
\begin{proof}
Equivalently, we show that the log evidence is unimodal.
\begin{align*}
    L &= \log p(\vx|d,\alpha) 
    \\
    &=
	\sum_{l=1}^{K-1} \log(\alpha + ld)
	+
	\sum_{i=1}^K \log\Gamma(n_i - d)
	+
	\log\Gamma(1 + \alpha)
	-
	K \log\Gamma(1 - d)
	-
	\log\Gamma(\alpha + N)
\end{align*}
It is sufficient to show that the partial derivative w.r.t. $d$ has at most one postive root.
\begin{align*}
    \D{L}{d}
    &=
	\sum_{l=1}^{K-1} \frac{l}{\alpha + ld}
	-
	\sum_{i=1}^K
	\left(
	    \psi_0(n_i - d)
	    -
	    \psi_0(1 - d)
	\right)
    \\
    &=
	\sum_{l=1}^{K-1} \frac{l}{\alpha + ld}
	+
	\sum_{i=1}^K
	\sum_{j=1}^{n_i-1}
	\frac{1}{d - j}
\end{align*}
Note that as $d \to 1$, the derivative tends to $-\infty$.
Combined with the observation that it is a linear combination of convex functions, there is at most one root for $\D{L}{d} = 0$.
\end{proof}

\begin{theorem}[Unimodal evidence on $\alpha$] \label{thm:unimodal:evidence:alpha}
    The evidence $p(\vx|d,\alpha)$ given by \eqref{eq:evidence} has only one local maximum (unimodal), on the region $\alpha>0$, for a fixed $d$.
\end{theorem}
\begin{proof}
Similar to theorem~\ref{thm:unimodal:evidence:d}, it is sufficient to show that the partial derivative w.r.t. $\alpha$ has at most one postive root.
\begin{align*}
    \D{L}{\alpha}
    &=
	\sum_{l=1}^{K-1} \frac{1}{\alpha + ld}
	+
	\psi_0(1 + \alpha)
	-
	\psi_0(\alpha + N)
    \\
    &=
	\sum_{l=1}^{K-1} \frac{1}{\alpha + ld}
	-
	\sum_{j=1}^{N-1} \frac{1}{j + \alpha}
\end{align*}
Let $\alpha = \frac{1}{x}$ be a root, then,
\begin{align}
    \sum_{i=1}^{K-1} \frac{1}{1 + xid} = \sum_{j=1}^{N-1} \frac{1}{1 + xj}.
\end{align}
Note that since $xid < xi$, $\frac{1}{1 + xid} > \frac{1}{1 + xi}$ for $1 \leq i \leq K-1$.
Therefore, we can split the equality as follows:
\begin{align}
    f_i(x) = a_i \frac{1}{1 + xid} &= \frac{1}{1 + xj} = g_i(x) \quad &\text{for $i \leq K-1$}\\
    f_{ij}(x) = b_{ij} \frac{1}{1 + xid} &= c_{ij} \frac{1}{1 + xj} = g_{ij}(x) \quad &\text{for $i \leq K-1$ and $K < j < N$}
\end{align}
where $0 \leq a_i, b_{ij}, c_{ij} \leq 1$,
$\forall i < K,\; a_i + \sum_j b_{ij} = 1$, and
$\forall j < N,\; \sum_i c_{ij} = 1$.
Fix $a_i, b_{ij}, c_{ij}$'s, and now suppose $\frac{1}{y} > \frac{1}{x} > 0$ is another positive root.
Then, we observe the following strict inequalities due to $0 \leq d < 1$,
\begin{align}
    \frac{f_i(x)}{f_i(y)} = \frac{1 + yid}{1 + xid} &< \frac{1 + yj}{1 + xj} = \frac{g_i(x)}{g_i(y)} \quad &\text{for $i \leq K-1$}\\
    \frac{f_{ij}(x)}{f_{ij}(y)} = \frac{1 + yid}{1 + xid} &< \frac{1 + yj}{1 + xj} = \frac{g_{ij}(x)}{g_{ij}(y)} \quad &\text{for $i \leq K-1$ and $K < j < N$}
\end{align}
Using lemma~\ref{lem:inequality:ratioSum} to put the sum back together, we obtain,
\begin{align}
    \sum_{i=1}^{K-1} \frac{1}{1 + yid} &> \sum_{j=1}^{N-1} \frac{1}{1 + yj}.
\end{align}
which is a contradiction to our assumption that $\frac{1}{y}$ is a positive root.
\end{proof}

\begin{lemma}\label{lem:inequality:ratioSum}
    If $f_j,g_j>0$, $f_j(x)=g_j(x)$ and $\frac{f_j(y)}{f_j(x)} > \frac{g_j(y)}{g_j(x)}$ for all $j$, then $\sum_j f_j(y)>\sum_j g_j(y)$.
\end{lemma}

\section*{Acknowledgments}
  We thank E.~J.~Chichilnisky, A.~M.~Litke, A.~Sher and J.~Shlens for
  retinal data, and Y.~W.~Teh and A.~Cerquetti for helpful comments on the manuscript.
  This work was supported by a Sloan Research Fellowship, McKnight
  Scholar's Award, and NSF CAREER Award IIS-1150186 (JP). Parts of
  this manuscript were presented at the Advances in Neural Information
  Processing Systems (NIPS) 2012 conference.

\end{document}